\documentclass[journal,10pt]{IEEEtran}

\usepackage{hyperref}
\usepackage{cite}
\usepackage{graphicx}
\usepackage{array}
\usepackage{mdwmath}
\usepackage{amssymb}
\usepackage{mdwtab}
\usepackage{stfloats}
\usepackage[tight,footnotesize]{subfigure}
\usepackage{amsmath,amsthm,amsfonts}
\usepackage{threeparttable}
\usepackage{color}
\usepackage{url}
\usepackage{algorithm}
\usepackage{algorithmic}
\usepackage{multicol}
\usepackage{epstopdf}
\usepackage{flushend}
\usepackage{epsfig}
\usepackage{xcolor}
\usepackage{textcomp}
\usepackage{bm}
\usepackage{setspace}
\usepackage{float}
\usepackage{bm}
\usepackage{microtype}
\usepackage{booktabs}
\usepackage{array}
\theoremstyle{definition}
\newtheorem{proposition}{Proposition}

\newtheorem{remark}{Remark}
\newtheorem{lemma}{Lemma}

\allowdisplaybreaks[4]

\ifCLASSINFOpdf

\else

\fi

\begin{document}
\title{\fontsize{22pt}{28pt}\selectfont Ultra-Low-Cost Hybrid Beamforming: A New Static-\\Connection Architecture with Sparse Phase-Shifter Sharing}

\author{

	Honghao Wang, Qingqing Wu, Yifei Wu, Yuxuan Chen, Wen Chen, and Derrick Wing Kwan Ng

\vspace{-20pt}

\thanks{H. Wang, Q. Wu, and W. Chen are with the School of Information Science and Electronic Engineering, Shanghai Jiao Tong University, Shanghai 200240, China (e-mail: hhwang@sjtu.edu.cn; qingqingwu@sjtu.edu.cn; wenchen@sjtu.edu.cn). Y. Wu is with the Institute for Digital Communications, Friedrich-Alexander-Universität Erlangen-Nürnberg (FAU), Germany (e-mail: yifei.wu@fau.de). Y. Chen is with the Huawei Device Company Ltd., Shenzhen 518129, China (e-mail: chenyuxuan56@huawei.com). D. W. K. Ng is with the School of Electrical Engineering and Telecommunications, University of New South Wales, Sydney, NSW 2052, Australia (e-mail: w.k.ng@unsw.edu.au). (\textit{Corresponding author: Qingqing Wu}).}
}

\maketitle

\begin{abstract}
Hybrid beamforming is a promising solution for high-frequency multi-antenna wireless systems, but its practical implementation remains significantly constrained by the hardware cost and complexity of analog phase-shifter (PS) networks. Although sub-connected architectures significantly simplify the analog network, their conventional realization still requires a dedicated PS for each antenna, which leads to considerable layout area, wiring complexity, calibration efforts, and control overhead. To effectively address this issue, this paper proposes a novel static-connection architecture with sparse PSs for ultra-low-cost sub-connected hybrid beamforming, where multiple antennas within each sub-array can share a common PS through an optimized yet fixed PS-to-antenna connection matrix. The proposed architecture preserves static hardware connections while enabling dynamic beam control via adaptive PS phase-shift adjustments and digital precoding. For the single-radio-frequency (RF)-chain scenario, the sparse-PS connection design is transformed into an antenna-grouping problem, for which structural properties are characterized analytically and an efficient algorithm is developed. For the multi-RF-chain scenario, we develop a competent quality-of-service (QoS)-majorization-minimization (MM) algorithm to handle the mixed discrete-continuous joint optimization problem. Numerical results demonstrate that the proposed architecture substantially reduces the PS count while preserving most of the beamforming capability provided by the traditional full-PS sub-connected architecture. In particular, the proposed design achieves PS-count reductions of $37.5\%$ and $62.5\%$ in single-RF-chain and multi-RF-chain systems, respectively, while effectively avoiding deep-null and grating-lobe degradations associated with rigid deterministic connection schemes. These results provide engineering insights into static sparse-PS sharing: the key to hardware-efficient hybrid beamforming is not merely reducing the PS count, but also preserving the essential analog-domain degrees of freedom through optimized PS connection topologies.
\end{abstract}

\vspace{-5pt}
\section{Introduction}
Future high-frequency wireless systems rely on directional transmission provided by large-scale antenna to compensate for severe propagation loss and to achieve communication services \cite{Elbir_BF_suv,WanSJ_HB_GNN,WangHH_low_altitude_Mag}. In this context, fully digital precoding becomes impractical since connecting each antenna directly to a dedicated radio frequency (RF)-chain, incurs excessive hardware costs, power consumption, and implementation burden \cite{Molisch_HB_suv,Albreem_BF_suv,WangHH_delay_MA}. As a remedy, hybrid beamforming has been proposed to address this challenge by exploiting fewer RF-chains combined with a high-dimensional analog phase-shifter (PS) network, enabling most spatial processing to be performed in the analog domain, while a low-dimensional digital precoder handles data multiplexing and interference management \cite{Sohrabi_HB_suv,YuXH_HB_switch_1}. In general, hybrid beamforming architectures can be classified into fully-connected and sub-connected structures \cite{GaoXY_HB_EE}. Specifically, each RF-chain in the former one is connected to all antennas through a dense network of PSs, splitters, couplers and crossings, which maximizes analog-domain degrees of freedom (DoFs) at the expense of extremely high implementation complexity, layout area, and insertion loss \cite{Ahmed_HB_suv,YuXH_HB_switch_2}. These inherent drawbacks become more pronounced at higher carrier frequencies and with increased hardware integration density. As a result, considerable attention has shifted to the sub-connected architecture, in which each RF-chain exclusively serves its own antenna sub-array. Compared to the fully-connected architecture, the sub-connected one features a simpler analog network and superior implementation scalability, thus becoming a promising solution for practical hybrid beamforming \cite{ZhangZL_HB_sub}.\looseness=-1

Nevertheless, the conventional sub-connected architecture still requires a dedicated PS for each antenna. In highly integrated high-frequency arrays, where the short wavelength pushes the array aperture toward increasingly compact implementations, this poses a critical hardware bottleneck: the PS count directly affects layout area, wiring density, package complexity, calibration overhead, and thermal management \cite{Bogale_HB_RF_PS}. Consequently, the question of how to substantially reduce the PS count without causing intolerable performance degradation, thereby lowering hardware cost and complexity, has naturally emerged as a fundamental obstacle in the practical commercialization of high-frequency array systems. Notably, a growing body of work has attempted to alleviate the reliance on high-precision PSs by incorporating switch networks \cite{Payami_HB_PS,Payami_HB_PS_switch}. Representative examples include fully-adaptive-connected structures and switch/dynamic-PS mixtures \cite{ZhuXD_HB_adaptive}, switch-aided designs leveraging low-resolution PSs \cite{LiHY_HB_low_resolution_PS_dynamic_array,GaoH_HB_low_reso_DAC}, switch-assisted architectures based on static PSs \cite{YuXH_HB_switch_2,Alkhateeb_HB_fixed_PS_switch,Bogale_HB_fixed_PS_switch}, and switch-only analog connections \cite{Ardah_HB_PS_switch,Méndez_HB_PS_switch}. In parallel, another line of research has presented the concept of dynamic sub-arrays, where the mapping between RF-chains and antennas is adaptively reconfigured according to varying service targets \cite{JinJN_HB_dynamic_array,Park_HB_dynamic_array}. However, the former group of schemes relies on dynamic switch networks, transforming analog-domain design into a highly complex combinatorial optimization problem \cite{Ratnam_HB_sele}. Thus, instead of re-optimizing the analog connection pattern for each target-dependent channel, it can be more practical to optimize a sparse-PS static-connection pattern offline. Moreover, at high-frequency carriers, dynamic signal rerouting further exacerbates issues such as parasitic coupling and active reflections. Meanwhile, the latter group of solutions requires complicated switch networks, additional control logic, and increased calibration overhead, while the non-static wiring also considerably complicates RF impedance matching \cite{HuYB_HB_sub}. Consequently, the aforementioned approaches essentially replace the original PS-related costs with alternative, and often more severe, hardware constraints and complexities, thereby compromising the inherent hardware simplicity that originally made the sub-connected architecture attractive. This directly contradicts the primary objective of reducing the PS count to develop a genuinely ultra-low-cost and highly implementable hybrid beamforming architecture.

Motivated by the above objective, engineering implementations often resort to deterministic PS-sharing strategies, where multiple antennas in the same sub-array are connected to one shared PS according to a predetermined pattern, thereby reducing the PS count and minimizing the occupied layout area \cite{Rupakula_HB_PS_num,Valle_HB_PS_num}. Such schemes feature hardware efficiency since they preserve a static-connection structure and eliminate the need for real-time reconfiguration. However, deterministic sharing strategies are inherently rigid: they do not exploit the spatial characteristics of user channels, and their resulting beampatterns may contain undesired deep nulls or strong grating lobes, which can severely degrade array gain in certain target directions \cite{Balanis_antenna}. Therefore, in these practical situations, the core technical challenge is no longer simply reducing the number of PSs, but rather optimizing the PS-sharing pattern itself, so that hardware is simplified without excessively sacrificing beam management capability and communication performance.\looseness=-1

To address this practical requirement, this paper investigates a novel static-connection architecture with sparse PSs for sub-connected hybrid beamforming. The key idea is that, within each sub-array, multiple antennas are allowed to share PSs, and the PS-to-antenna connection pattern is optimized offline according to the spatial characteristics of target user channels and then kept fixed, while the PS phase shifts and digital precoder are adjusted according to the instantaneous transmission task. Compared to conventional fully-connected and sub-connected PS networks, the proposed architecture substantially reduces the PS count and further simplifies the associated hardware implementation. Unlike architectures employing switches or dynamic sub-arrays, e.g., \cite{Payami_HB_PS,Payami_HB_PS_switch,ZhuXD_HB_adaptive,LiHY_HB_low_resolution_PS_dynamic_array,GaoH_HB_low_reso_DAC,Alkhateeb_HB_fixed_PS_switch,Bogale_HB_fixed_PS_switch,Ardah_HB_PS_switch,Méndez_HB_PS_switch,JinJN_HB_dynamic_array,Park_HB_dynamic_array}, it avoids real-time reconfiguration of connections, along with the associated switching, control, and routing overhead. Furthermore, compared to deterministic PS-sharing schemes, this specifically designed sparse-PS static-connection pattern preserves superior beam control capability and enhanced communication performance without introducing extra hardware cost and complexity. From an engineering implementation perspective, the proposed architecture is particularly appealing when hardware cost, array integration, board-level area, control-line fan-out, wiring complexity, analog-front-end stability, and control overhead are the primary design considerations \cite{Bogale_HB_RF_PS}. Indeed, it transforms the PS connection pattern from merely being a tool for hardware simplification into a new architectural DoF, shifting the design emphasis from instantaneous analog-domain reconfiguration to strategic offline hardware configuration. Such a static and preconfigured design philosophy is highly compatible with practical manufacturing requirements and constraints.

In scenarios such as fixed wireless access \cite{Aldubaikhy_fixed_access}, road or railway coverage \cite{LiJ_train}, industrial Internet-of-Things deployments \cite{Jabbar_industry}, and other directional service environments \cite{TanJR_V2X}, user directions typically cluster within a narrow angular region and thus can be efficiently represented by a finite set of candidate service directions \cite{Giordani_beam_3GPP}. The proposed architecture is particularly well-suited to these settings, since it leverages the spatial channel characteristics of all potential directions to optimize the sparse-PS static-connection pattern offline, thereby significantly reducing hardware cost and implementation complexity. Meanwhile, by optimizing the high-dimensional PS connection network offline instead of frequently reconfiguring it according to instantaneous traffic, the architecture greatly reduces computational complexity and control overhead. This design approach also aligns well with practical beam management systems, where finite beam/codebook structures \cite{Alkhateeb_codebook}, hotspot distributions \cite{Qurratulain_Beam}, or sectorized service regions \cite{XueQ_Beam}, facilitate efficient beam control in realistic deployment scenarios \cite{Alrabeiah_Beam}. The main contributions are summarized as follows:
\vspace{-1pt}
\begin{itemize}
\item This paper proposes an ultra-low-cost hybrid beamforming framework enabled by a static-connection architecture featuring a reduced PS count. Specifically, the PS connection matrix is optimized offline and kept fixed, allowing multiple antennas within each sub-array to effectively share PS resources and thereby integrating static hardware connections with dynamic beam control.
\item For the single-RF-chain scenario, we convert the sparse-PS connection design into an antenna-grouping problem and characterize its structural properties. Moreover, we further derive analytical results and unveil practical engineering insights for the single-beam case, and develop an efficient connection optimization algorithm applicable to both single-beam and beam-switching cases.
\item For the multi-RF-chain scenario, we develop an efficient quality-of-service (QoS)-majorization-minimization (MM) algorithm for the resulting mixed discrete-continuous joint optimization problem. It alternately updates the digital precoder by a closed-form uplink--downlink-duality-based solution, the PS phase shifts by a closed-form unit-modulus MM solution, and the binary PS connection and assignment matrices by solving a mixed-integer linear program (MILP) built from linear MM surrogates.
\item Numerical results demonstrate that the proposed architecture can substantially reduce PS-related hardware cost while preserving nearly full beamforming capability of the conventional sub-connected counterpart. Specifically, the proposed design achieves PS-count reductions of $37.5\%$ and $62.5\%$ for single-RF-chain and multi-RF-chain systems, respectively, and effectively avoids the undesirable deep-null and grating-lobe issues arising from rigid deterministic connection patterns.
\end{itemize}

The remainder of this paper is organized as follows. Section \ref{sec2} presents the system model and formulates the design as a mixed discrete-continuous optimization problem. Section \ref{sec3} investigates the single-RF-chain scenario, developing the analytical results and an efficient algorithm. Section \ref{sec4} extends the design to the multi-RF-chain scenario and proposes an effective QoS-MM algorithm. Section \ref{sec5} provides numerical results, and Section \ref{sec6} concludes the paper.\looseness=-1

\emph{Notations:} Scalars, vectors, matrices, and sets are denoted by italic, bold lowercase, bold uppercase, and calligraphic letters, respectively. For a set $\mathcal{L}$, $|\mathcal{L}|$ is its cardinality, and $\mathcal{L}\setminus\{l\}$ denotes the set obtained by excluding $l$ from $\mathcal{L}$. Here, $\mathbb{C}^{M \times N}$, $\{0,1\}^{M \times N}$, $\mathbb{N}$, and $\mathbb{R}^{+}$ denote the spaces of $M \times N$ complex and binary matrices, and the sets of natural and positive real numbers, respectively. For a real scalar $x$, $\lfloor x\rfloor$ is the largest integer not exceeding $x$. For a complex scalar $x$, $|x|$, $\Re\{x\}$, and $\Im\{x\}$ denote its modulus, real part, and imaginary part, respectively. For a vector $\mathbf{x}$, $\mathbf{x}^*$, $\mathbf{x}^T$, $\mathbf{x}^H$, $\|\mathbf{x}\|_p$, $\operatorname{diag}(\mathbf{x})$, $\arg(\mathbf{x})$, $\mathbf{x}[i]$, and $\mathbf{x}[M\!\!\!:\!\!\!N]$ denote its conjugate, transpose, conjugate transpose, $p$-norm, diagonalization matrix, element-wise phases, $i$-th entry, and the subvector of its $M$-th through $N$-th entries, respectively. For a matrix $\mathbf{A}$, $\mathbf{A}[i,j]$, $\mathbf{A}[i\,,:\,]$, $\mathbf{A}[\,:\,,j]$, $\operatorname{tr}(\mathbf{A})$, $\mathbf{A}^{-1}$, $\operatorname{vec}(\mathbf{A})$, and $\lambda_{\max}(\mathbf{A})$ denote its $(i,j)$-th entry, $i$-th row, $j$-th column, trace, inverse, vectorization, and maximum eigenvalue, respectively, while $\mathbf{A}\succeq\mathbf{0}$ and $\mathbf{A}\succ\mathbf{0}$ indicate that $\mathbf{A}$ is positive semidefinite and positive definite, respectively. $\mathbf{I}^N$ and $\mathbf{0}^{M \times N}$ denote the $N\times N$ identity matrix and $M\times N$ zero matrix, respectively. $\mathcal{CN}(\mu,\sigma)$ denotes the circularly symmetric complex Gaussian distribution with mean $\mu$ and covariance $\sigma$.\looseness=-1

\begin{figure}[t]
	\centering
	\includegraphics[width=2.82in]{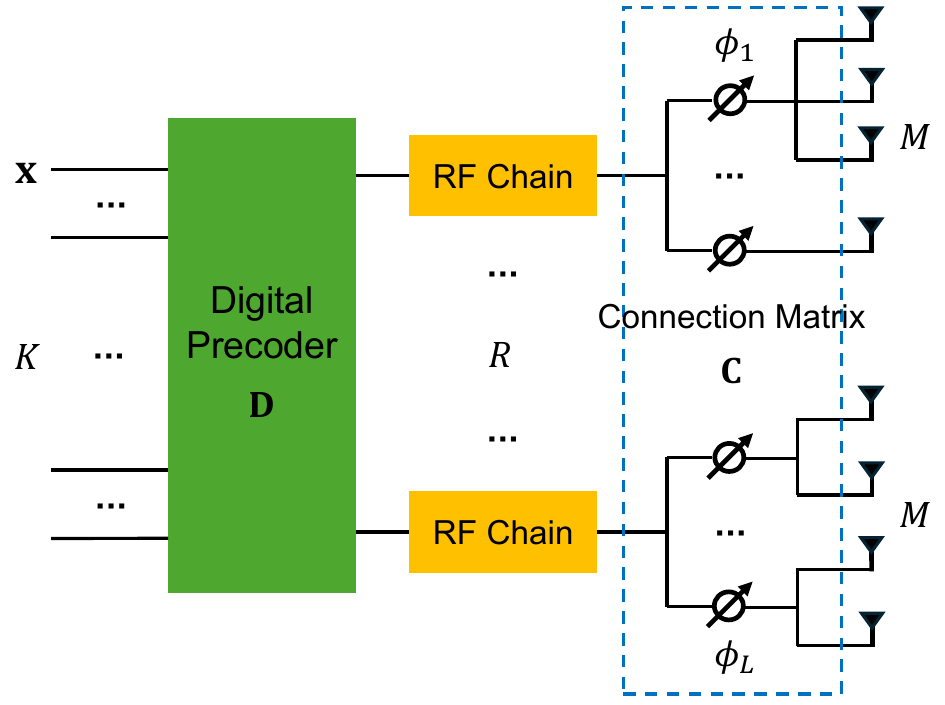}
	\vspace{-10pt}
	\caption{Illustration of the new static-connection architecture with sparse PSs for sub-connected hybrid beamforming.}
	\label{Fig1}
	\vspace{-15pt}
\end{figure}

\vspace{-5pt}
\section{System Model}\label{sec2}
Consider a multiuser multiple-input single-output (MU-MISO) downlink communication system, where a base station (BS) employs a uniform linear array (ULA) with $N$ half-wavelength-spaced antennas and $R$ radio-frequency (RF) chains, serving $K$ candidate service directions (or, equivalently, $K$ candidate single-antenna users), indexed by $\mathcal{K}\triangleq\{1,\ldots,K\}$. For notational convenience, these candidate directions are still referred to as ``users'' in the sequel. We adopt a sub-connected hybrid beamforming architecture, where each RF-chain drives a dedicated sub-array consisting of $M$ adjacent antennas, thereby satisfying $N=RM$. To reduce hardware complexity, multiple antennas within the same sub-array share PSs, as illustrated in Fig. \ref{Fig1}. Let $L$ denote the total number of available PSs, where $R\le L\le N$. Define the antenna index set as $\mathcal{N}\triangleq\{1,\ldots,N\}$, the RF-chain/sub-array index set as $\mathcal{R}\triangleq\{1,\ldots,R\}$, and the PS index set as $\mathcal{L}\triangleq\{1,\ldots,L\}$. The received signal at user $k$ can be expressed as\looseness=-1
\begin{equation}\label{eq_signal_k}
	y_k=\mathbf{h}_k^H\mathbf{w}_kx_k+\sum\nolimits_{i\in\mathcal{K},i\ne k}\mathbf{h}_k^H\mathbf{w}_ix_i+z_k,~k\in\mathcal{K},
\end{equation}
where $\mathbf{h}_k,\mathbf{w}_k\in\mathbb{C}^{N\times1}$, and $x_k,z_k\in\mathbb{C}$ are the channel vector, transmit beamforming vector, information symbol with normalized power, and additive white Gaussian noise with $z_k\sim\mathcal{CN}(0,\sigma_k^2)$ of user $k$, respectively. In this paper, we focus on the line-of-sight (LoS) channel model following the form
\vspace{0pt}
\begin{equation}\label{eq_channel}
	\mathbf{h}_k=\sqrt{\alpha_k}\mathbf{a}(\theta_k),~k\in\mathcal{K},
\end{equation}
where $\alpha_k\in\mathbb{R}^+$ is the channel power gain of user $k$ and $\mathbf{a}(\theta_k)\triangleq\big[1,e^{-j\psi\sin(\theta_k)},\ldots,e^{-j(N-1)\psi\sin(\theta_k)}\big]^T\in\mathbb{C}^{N\times 1}$ is the ULA steering vector associated with user $k$'s angle-of-departure (AoD) $\theta_k$. $\psi=2\pi{d}/\lambda\overset{d=\lambda/2}{=}\pi$ is a constant value, with $\lambda$ representing the carrier wavelength and $d=\lambda/2$ being the half-wavelength antenna spacing.

We adopt a standard hybrid analog-digital beamforming structure, expressed as $\mathbf{w}_k=\mathbf{F}\mathbf{d}_k$, where $\mathbf{F}\in\mathbb{C}^{N\times{R}}$ is the analog beamforming matrix implemented by PSs and its associated connection pattern, and $\mathbf{d}_k\in\mathbb{C}^{R\times1}$ is the digital precoder of user $k$ at baseband. Due to the nature of sub-connected architecture, $\mathbf{F}$ is block diagonal:
\begin{equation}
	\mathbf{F}=\left[\begin{matrix}
		\mathbf{f}_1 & \ldots & \mathbf{0}\\[-5pt]
		\vdots & \ddots & \vdots\\[-3pt]
		\mathbf{0} & \ldots & \mathbf{f}_R
	\end{matrix}\right],
\end{equation}
where $\mathbf{f}_r\in\mathbb{C}^{M\times1}$ is the analog beamformer for the $r$-th sub-array. In practice, assigning one dedicated PS per antenna may be unnecessary and hardware-expensive. In fact, when users are concentrated within a limited set of direction angles, a fully flexible hybrid beamforming becomes redundant, and reducing the PS count may cause only a minor performance loss while substantially lowering implementation costs. Moreover, in highly integrated high-frequency arrays, various implementation constraints, such as limited layout area, wiring resources, and packaging space, often make per-antenna PS deployment difficult, if not impossible. Therefore, it is of practical interest to develop a sparse-PS architecture that achieves a more favorable trade-off between communication performance and hardware efficiency. To substantially reduce the PS count, we propose connecting multiple antennas in the same sub-array to a shared PS, while keeping the connection pattern fixed. Specifically, the analog beamforming matrix can be modeled as\footnote{To focus on the beamforming-level design of sparse-PS static connections, the present model assumes ideal lossless PS sharing and equal-amplitude distribution among antennas within the same sub-array.\vspace{-3.4pt}}\looseness=-1
\begin{equation}\label{eq_F}
	\mathbf{F}=\frac{1}{\sqrt{M}}\mathbf{C}\mathbf{\Phi}=\frac{1}{\sqrt{M}}\left[\begin{matrix}
		\mathbf{C}_1 & \ldots & \mathbf{0}\\[-5pt]
		\vdots & \ddots & \vdots\\[-3pt]
		\mathbf{0} & \ldots & \mathbf{C}_R
	\end{matrix}\right]\hspace{-3pt}\left[\begin{matrix}
	\bm{\phi} & \ldots & \mathbf{0}\\[-5pt]
	\vdots & \ddots & \vdots\\[-3pt]
	\mathbf{0} & \ldots & \bm{\phi}
	\end{matrix}\right],
\end{equation}
where $\mathbf{C}\in\{0,1\}^{N\times{RL}}$ is a block-diagonal binary connection matrix describing PS-to-antenna wiring, and $\mathbf{\Phi}\in\mathbb{C}^{RL\times{R}}$ is the block-diagonal phase-shift matrix with unit-modulus entries. Let $\mathbf{C}_r\in\{0,1\}^{M\times L}$ denote the $r$-th diagonal submatrix in $\mathbf{C}$, where $\mathbf{C}_r[m,l]=1$ indicates that the $l$-th PS is connected to the $m$-th antenna in the $r$-th sub-array. $\bm{\phi}\triangleq[\phi_1,\ldots,\phi_L]^T\in\mathbb{C}^{L\times1}$ collects the PS phase shifts, where $\phi_l$ is the phase shift of the $l$-th PS with $|\phi_l|=1$. Then, we have\looseness=-1
\begin{equation}
	\mathbf{f}_r=\frac{1}{\sqrt{M}}\mathbf{C}_r\bm{\phi},~r\in\mathcal{R}.
\end{equation}

To explicitly model the fact that each PS can be connected to antennas within only one sub-array, we introduce an auxiliary binary PS-to-RF assignment matrix $\mathbf{S}\in\{0,1\}^{R\times L}$, where $\mathbf{S}[r,l]=1$ indicates that the $l$-th PS is assigned to the $r$-th RF-chain, or equivalently, to the $r$-th sub-array. The connection constraints impose the following conditions: i) each antenna must be connected to exactly one PS; ii) each PS must be assigned to exactly one RF-chain/sub-array; iii) an antenna can only be connected to a PS assigned to its own RF-chain/sub-array; and iv) each PS must be connected to at least one antenna in its assigned sub-array, which are formulated as
\begin{align}
	&\sum\nolimits_{l\in\mathcal{L}}\mathbf{C}_r[m,l]=1,~m\in\mathcal{M}\triangleq\{1,\ldots,M\},~r\in\mathcal{R},\label{eq_PS_allocation_1}\\
	&\sum\nolimits_{r\in\mathcal{R}}\mathbf{S}[r,l]=1,~l\in\mathcal{L},\label{eq_PS_allocation_2}\\
	&\mathbf{C}_r[m,l]\le \mathbf{S}[r,l],~m\in\mathcal{M},~l\in\mathcal{L},~r\in\mathcal{R},\label{eq_PS_allocation_3}\\
	&\mathbf{S}[r,l]\le \sum\nolimits_{m\in\mathcal{M}}\mathbf{C}_r[m,l],~l\in\mathcal{L},~r\in\mathcal{R}.\label{eq_PS_allocation_4}
\end{align}
Specifically, \eqref{eq_PS_allocation_1} and \eqref{eq_PS_allocation_2} guarantee conditions i) and ii), respectively. Constraint \eqref{eq_PS_allocation_3} couples PS-to-antenna connection with PS-to-RF assignment, ensuring condition iii). Constraint \eqref{eq_PS_allocation_4} rules out idle PS assignments by holding condition iv). Together, \eqref{eq_PS_allocation_2}--\eqref{eq_PS_allocation_4} make $\mathbf{S}[r,l]$ an exact indicator of whether the $l$-th PS is used by the $r$-th sub-array.


With each RF-chain ultimately connected to $M$ antennas in its corresponding sub-array, the factor $1/\sqrt{M}$ is chosen as the power normalization factor that ensures $\|\mathbf{f}_r\|_2=1$. Since the sub-arrays are disjoint, we have $\mathbf{F}^H\mathbf{F}=\mathbf{I}^R$, and thus
\begin{equation}
	\sum\nolimits_{k\in\mathcal{K}}\|\mathbf{w}_k\|_2^2=\sum\nolimits_{k\in\mathcal{K}}\|\mathbf{F}\mathbf{d}_k\|_2^2=\sum\nolimits_{k\in\mathcal{K}}\|\mathbf{d}_k\|_2^2.
\end{equation}

With $\mathbf{w}_k=\mathbf{F}\mathbf{d}_k=\mathbf{C}\mathbf{\Phi}\mathbf{d}_k/\sqrt{M}$, the signal-to-interference-plus-noise ratio (SINR) at user $k$ can be expressed as
\begin{equation}
	\mathrm{SINR}_k=\frac{|\mathbf{h}_k^H\mathbf{C}\mathbf{\Phi}\mathbf{d}_k|^2}{\sum_{i\in\mathcal{K},i\ne{k}}|\mathbf{h}_k^H\mathbf{C}\mathbf{\Phi}\mathbf{d}_i|^2+M\sigma_k^2},~k\in\mathcal{K}.
\end{equation}

We aim to minimize the transmit power while satisfying the QoS requirements by jointly optimizing $\left\{\{\mathbf{d}_k\}_{k\in\mathcal{K}},\mathbf{C},\mathbf{S},\bm{\phi}\right\}$. Accordingly, the optimization problem is formulated as
\begin{subequations}
	\begin{alignat}{2}
		\mathrm{(P1)}:~&\underset{\{\mathbf{d}_k\}_{k\in\mathcal{K}},\mathbf{C},\mathbf{S},\bm{\phi}}{\min}~&&\sum\nolimits_{k\in\mathcal{K}}\|\mathbf{d}_k\|_2^2\label{eq_P1_a}\\
		&\hspace{-24pt}\mathrm{s.t.}&&\hspace{-47pt}\mathrm{SINR}_k\ge\gamma_k,~k\in\mathcal{K},\label{eq_P1_b}\\
		&&&\hspace{-47pt}\mathbf{C}_r[m,l]\in\{0,1\},~m\in\mathcal{M},~l\in\mathcal{L},~r\in\mathcal{R},\label{eq_P1_c}\\
		&&&\hspace{-47pt}\mathbf{S}[r,l]\in\{0,1\},~r\in\mathcal{R},~l\in\mathcal{L},\label{eq_P1_d}\\
		&&&\hspace{-47pt}|\phi_l|=1,~l\in\mathcal{L},\label{eq_P1_e}\\
		&&&\hspace{-47pt}\eqref{eq_PS_allocation_1}-\eqref{eq_PS_allocation_4}.\nonumber
	\end{alignat}
\end{subequations}
The resulting optimization problem is a challenging mixed-integer nonlinear program (MINLP), due to the binary variables in $\{\mathbf{C},\mathbf{S}\}$ and the intractable coupling between discrete and continuous variables. To gain analytical insights and derive efficient algorithms, we next investigate two specific scenarios: i) the single-RF-chain scenario ($R=1$), and ii) the general multi-RF-chain scenario ($R>1$).

\begin{remark}
	The parameter $K$ can be interpreted as the cardinality of the set of candidate service users/directions, rather than the number of users that must be active simultaneously in every transmission interval. For the purpose of reducing hardware cost and implementation complexity, the sparse-PS connection proposed in this paper is kept static, i.e., the PS-to-antenna connection pattern is determined offline and is not reconfigured according to the instantaneous active-user set. Consequently, the PS network should be provisioned against the worst-case communication condition at the design stage, which motivates problem (P1) by assuming that all candidate users in $\mathcal{K}$ are simultaneously active. The resulting design naturally accommodates practical scenarios in which only a subset $\mathcal{S}\subseteq\mathcal{K}$ of users is active. Specifically, let $\{\mathbf{w}_k^\star\}_{k\in\mathcal{K}}$ be a feasible beamforming solution to (P1). If only users in $\mathcal{S}$ are active, we can simply set $\mathbf{w}_i=\mathbf{0},~\forall i\notin\mathcal{S}$, while keeping the remaining variables unchanged. Then, for any $k\in\mathcal{S}$,
	\begin{align}
		\mathrm{SINR}_k^{(\mathcal{S})}&=\frac{|\mathbf{h}_k^H \mathbf{w}_k^\star|^2}{\sum_{i\in\mathcal{S},i\neq k}|\mathbf{h}_k^H \mathbf{w}_i^\star|^2+\sigma_k^2}\nonumber\\[-3pt]
		&\ge\frac{|\mathbf{h}_k^H \mathbf{w}_k^\star|^2}{\sum_{i\in\mathcal{K},i\neq k}|\mathbf{h}_k^H\mathbf{w}_i^\star|^2+\sigma_k^2}=\mathrm{SINR}_k^{(\mathcal{K})}.
	\end{align}
	Hence, any beamforming design that is feasible for this worst-case fully-active problem is automatically feasible for any active subset. Moreover, since interference is reduced when inactive users are excluded, the minimum transmit power for a subset-active case is no larger, and is generally smaller, than that required by the worst-case fully-active problem (P1).
\end{remark}

\vspace{-8pt}
\section{Single RF-Chain}\label{sec3}
In this section, we examine the special case of a single-RF-chain configuration ($R=1$). In this case, hybrid beamforming simplifies to analog beamforming facilitated by a sparse-PS connection network, allowing for deeper analytical insights into the problem structure. Since a single RF-chain can serve at most one active user at a time, the original worst-case design problem naturally becomes a beam-switching design: we seek a single PS-to-antenna connection matrix that supports different users/directions via reconfiguring the PS phase shifts. Accordingly, the transmit beamformer for user $k$ is written as%
\begin{equation}
	\mathbf{w}_k=\sqrt{\bar{P}_k}\mathbf{\bar{f}}_k,~k\in\mathcal{K},
\end{equation}
where $\bar{P}_k$ and $\mathbf{\bar{f}}_k$ are the allocated transmit power and unit-norm analog beamformer for user $k$, respectively. Given the sparse-PS connection model, we have
\begin{equation}
	\mathbf{\bar{f}}_k=\frac{1}{\sqrt{N}}\mathbf{\bar{C}}\bm{\bar{\phi}}_k,~k\in\mathcal{K},
\end{equation}
where $\mathbf{\bar{C}}\in\{0,1\}^{N\times{L}}$ represents the single-RF-chain PS connection matrix, $\bm{\bar{\phi}}_k\triangleq\big[\bar{\phi}_{k,1},\ldots,\bar{\phi}_{k,L}\big]^T\in\mathbb{C}^{L\times1}$, and $\bar{\phi}_{k,l}$ is the phase shift of the $l$-th PS for user $k$ with $\big|\bar{\phi}_{k,l}\big|=1$. Furthermore, the connection constraints \eqref{eq_PS_allocation_1}--\eqref{eq_PS_allocation_4} are reduced to
\begin{align}
	&\sum\nolimits_{l\in\mathcal{L}}\mathbf{\bar{C}}[n,l]=1,~n\in\mathcal{N},\label{eq_P2_e}\\
	&\sum\nolimits_{n\in\mathcal{N}}\mathbf{\bar{C}}[n,l]\ge1,~l\in\mathcal{L}.\label{eq_P2_f}
\end{align}
This configuration enables beam switching across users in a time-division manner. In particular, $\bar{\mathbf{C}}$ remains fixed as static wiring, while $\big\{\bm{\bar{\phi}}_k\big\}_{k\in\mathcal{K}}$ are adjusted to steer different beams toward different users. With the above structure, the received signal when serving user $k$ is given by
\begin{align}
	\bar{y}_k=\sqrt{\frac{\bar{P}_k}{N}}\mathbf{h}_k^H\mathbf{\bar{C}}\bm{\bar{\phi}}_kx_k+z_k,~k\in\mathcal{K},
\end{align}
yielding a signal-to-noise ratio (SNR) of the form
\begin{equation}
	\mathrm{SNR}_k=\frac{\bar{P}_k\big|\mathbf{h}_k^H\mathbf{\bar{C}}\bm{\bar{\phi}}_k\big|^2}{N\sigma_k^2},~k\in\mathcal{K}.
\end{equation}
Therefore, the single-RF-chain beam-switching design can be explicitly reformulated as the following optimization problem:\looseness=-1
\begin{subequations}
	\begin{alignat}{2}
		\mathrm{(P2)}:~&\underset{\mathbf{\bar{C}},\left\{\bm{\bar{\phi}}_k,\bar{P}_k\right\}_{k\in\mathcal{K}}}{\min}~&&\sum\nolimits_{k\in\mathcal{K}}\bar{P}_k\label{eq_P2_a}\\[-4pt]
		&\quad\quad\mathrm{s.t.}&&\hspace{-1pt}\mathrm{SNR}_k\ge\gamma_k,~k\in\mathcal{K},\label{eq_P2_b}\\
		&&&\hspace{-1pt}\mathbf{\bar{C}}[n,l]\in\{0,1\},~n\in\mathcal{N},~l\in\mathcal{L},\label{eq_P2_c}\\
		&&&\hspace{-1pt}\big|\bar{\phi}_{k,l}\big|=1,~k\in\mathcal{K},~l\in\mathcal{L},\label{eq_P2_d}\\
		&&&\hspace{-1pt}\eqref{eq_P2_e},~\eqref{eq_P2_f}.\nonumber
	\end{alignat}
\end{subequations}
This formula emphasizes the key design objective for single-RF-chain optimization: selecting a single connection matrix that effectively preserves array gain over a set of candidate users, so that per-user phase-shift reconfiguration can generate the required directional beams. Next, we start with a single-beam scenario to extract analytical insights, and extend to devise an efficient algorithm for the beam-switching scenario.\looseness=-1

\subsection{Single Beam with Dedicated PS Connection}\label{subsec3_1}
We first consider the single-beam scenario (i.e., $K=1$), which allows us to isolate and explicitly examine the fundamental impact of antenna grouping determined by the connection matrix $\mathbf{\bar{C}}$. By omitting the user index $k$ for brevity, problem (P2) is equivalent to the following array-gain maximization problem\footnote{Once the optimal solution $\big\{\mathbf{\bar{C}}^\star,\bm{\bar{\phi}}^\star\big\}$ of problem (P2.1) is derived, the minimum transmit power can be determined by $\bar{P}^\star=\frac{\gamma N\sigma^2}{\alpha\left|\mathbf{a}(\theta)^H\mathbf{\bar{C}}^\star\bm{\bar{\phi}}^\star\right|^2}$.}\looseness=-1
\begin{subequations}
	\begin{alignat}{2}
		\mathrm{(P2.1)}:\quad&\underset{\mathbf{\bar{C}},\bm{\bar{\phi}}}{\max}\quad&&\big|\mathbf{a}(\theta)^H\mathbf{\bar{C}}\bm{\bar{\phi}}\big|\label{eq_P2_1_a}\\[-3pt]
		&~\mathrm{s.t.}&&\big|\bar{\phi}_l\big|=1,~l\in\mathcal{L},\label{eq_P2_1_b}\\
		&&&\eqref{eq_P2_e},~\eqref{eq_P2_f},~\eqref{eq_P2_c}.\nonumber
	\end{alignat}
\end{subequations}
The objective function can be further rewritten as
\begin{align}\label{eq_SRSU_obj}
	\big|\mathbf{a}(\theta)^H\mathbf{\bar{C}}\bm{\bar{\phi}}\big|&=\left|\sum\nolimits_{l\in\mathcal{L}}\bar{\phi}_l\sum\nolimits_{n\in\mathcal{N}}\mathbf{\bar{C}}[n,l]e^{j(n-1)\pi\sin(\theta)}\right|\nonumber\\
	&\hspace{-8pt}=\left|\sum\nolimits_{l\in\mathcal{L}}\!\bar{\phi}_l\sum\nolimits_{n\in\mathcal{N}_l}\!\!\exp(j(n\!-\!1)\pi\sin(\theta))\right|,
\end{align}
where we denote $\mathcal{N}_l$, $l\in\mathcal{L}$, as the group of antennas connected to the $l$-th PS, i.e., $\mathcal{N}_l\triangleq\big\{n\in\mathcal{N}~|~\mathbf{\bar{C}}[n,l]=1\big\}$. For any fixed grouping induced by $\mathbf{\bar{C}}$, the optimal phase shift of each PS is obtained by coherently aligning the summed contributions of the antennas connected to that PS, i.e.,
\begin{align}
	\bar{\phi}_l&=\frac{\sum_{n\in\mathcal{N}_l}\mathbf{a}(\theta)[n]}{\big|\sum_{n\in\mathcal{N}_l}\mathbf{a}(\theta)[n]\big|}=\frac{\sum_{n\in\mathcal{N}_l}e^{-j(n-1)\pi\sin(\theta)}}{\big|\sum_{n\in\mathcal{N}_l}e^{-j(n-1)\pi\sin(\theta)}\big|}\label{eq_SRSU_phi_l_1}\\
	&=\exp\Big(\!\!-\!j\arg\Big(\sum\nolimits_{n\in\mathcal{N}_l}\!e^{j(n-1)\pi\sin(\theta)}\Big)\Big),~l\in\mathcal{L}\label{eq_SRSU_phi_l_2}.
\end{align}
Substituting this closed-form optimal phase-shift solution back into \eqref{eq_SRSU_obj} reduces the design problem to a pure antenna grouping problem with respect to (w.r.t.) $\mathbf{\bar{C}}$, revealing that performance is governed by the phase consistency among the antennas connected to the same PS. We have
\begin{equation}\label{eq_SRSU_gain}
	\big|\mathbf{a}(\theta)^H\mathbf{\bar{C}}\bm{\bar{\phi}}\big|=\sum\nolimits_{l\in\mathcal{L}}\!\left|\sum\nolimits_{n\in\mathcal{N}_l}\!\exp(j(n\!-\!1)\pi\sin(\theta))\right|.
\end{equation}
Specifically, after applying the optimal per-group phase rotation, the array gain becomes the sum of magnitudes of phasor sums within each group. Hence, the objective is to group antennas such that, in the desired AoD $\theta$, the corresponding steering phases are as consistent as possible within each group. This interpretation directly links sparse-PS connection to array gain loss: if a group contains antennas with widely separated steering phases, their contributions partially cancel and the resulting gain decreases. In the following, we investigate two special cases to develop analytical insights and then characterize the relationship between the PS count and the maximum achievable array gain in the general case.

\subsubsection{Case 1} Assume the number of antennas is divisible by the PS count, i.e., $N/L\in\mathbb{N}^+$, and the steering-phase variation across the antenna array is smaller than one full cycle, i.e.,
\begin{equation}
	|(N-1)\pi\sin(\theta)|<2\pi~\Rightarrow~|\sin(\theta)|<\frac{2}{N-1}.\label{eq_SRSU_case1_sin}
\end{equation}
Under these conditions, the steering phase varies monotonically with the antenna index and spans less than one full cycle, i.e., $2\pi$. Therefore, motivated by steering-phase consistency, we consider an ordered grouping with $L$ equally sized antenna groups, which offers an analytically tractable structure and reveals valuable engineering insights. Particularly, if the following condition holds, the first group $\mathcal{N}_1$ should contain the antenna indices $\{N,1,2,\ldots,N/L-1\}$:
\begin{align}
	&2\pi-|(N-1)\pi\sin(\theta)|<|\pi\sin(\theta)|,\nonumber\\
	\Rightarrow~&|\sin(\theta)|>\frac{2}{N}~\overset{\mathrm{(a)}}{\Rightarrow}~\frac{2}{N}<|\sin(\theta)|<\frac{2}{N-1}.\label{eq_SRSU_case1_sin_1}
\end{align}
where (a) follows from \eqref{eq_SRSU_case1_sin}. Otherwise, if
\begin{equation}
	|\sin(\theta)|\le\frac{2}{N},\label{eq_SRSU_case1_sin_2}
\end{equation}
the first group $\mathcal{N}_1$ should instead contain the antenna indices $\{1,2,\ldots,N/L\}$.

The remaining antennas are then sequentially assigned to the groups $\mathcal{N}_l$, $l\ge2$, in ascending order of their indices. Specifically, under \eqref{eq_SRSU_case1_sin_1}, the groups are given by $\mathcal{N}_l=\{N(l-1)/L,N(l-1)/L+1,\ldots,Nl/L-1\}$, $l\ge2$, whereas under \eqref{eq_SRSU_case1_sin_2}, they are given by $\mathcal{N}_l=\{N(l-1)/L+1,N(l-1)/L+2,\ldots,Nl/L\}$, $l\ge2$. Based on the above antenna grouping strategy, the optimal phase shift for the $l$-th PS connected to $\mathcal{N}_l$ is obtained by normalizing the coherent sum of the corresponding steering-vector entries:
\begin{align}
	\hspace{-3pt}\bar{\phi}_l&\overset{\mathrm{(b)}}{=}\frac{\sum_{i=2}^{N/L}e^{-j(i-2)\pi\sin(\theta)}+e^{-j(N-1)\pi\sin(\theta)}}{\left|\sum_{i=2}^{N/L}e^{-j(i-2)\pi\sin(\theta)}+e^{-j(N-1)\pi\sin(\theta)}\right|},~l=1\nonumber\\[-2pt]
	&\overset{\mathrm{(b)}}{=}\frac{\sum_{i=1}^{N/L}e^{-j(N(l-1)/L+i-2)\pi\sin(\theta)}}{\left|\sum_{i=1}^{N/L}e^{-j(N(l-1)/L+i-2)\pi\sin(\theta)}\right|},~l\ge2,~\text{or}\label{eq_phi_case1_b}\\[-2pt]
	&\overset{\mathrm{(c)}}{=}\frac{\sum_{i=1}^{N/L}e^{-j(N(l-1)/L+i-1)\pi\sin(\theta)}}{\left|\sum_{i=1}^{N/L}e^{-j(N(l-1)/L+i-1)\pi\sin(\theta)}\right|},~l\in\mathcal{L},\label{eq_phi_case1_c}
\end{align}
where (b) and (c) hold under \eqref{eq_SRSU_case1_sin_1} and \eqref{eq_SRSU_case1_sin_2}, respectively. By exploiting the structure of steering vector and properties of complex exponential summation, \eqref{eq_phi_case1_b} can be further simplified as\looseness=-1
\begin{equation}\label{eq_phi_case1}
	\bar{\phi}_l=\exp\Big(\!-\!j\,\Big(\,\frac{N}{L}l-\frac{N}{2L}-\frac{3}{2}\,\Big)\,\pi\sin(\theta)\Big),~l\ge2.
\end{equation}
Similarly, \eqref{eq_phi_case1_c} yields the following closed-form phase shift:
\begin{equation}\label{eq_phi_case2}
	\bar{\phi}_l=\exp\Big(\!-\!j\,\Big(\,\frac{N}{L}l-\frac{N}{2L}-\frac{1}{2}\,\Big)\,\pi\sin(\theta)\Big),~l\in\mathcal{L}.
\end{equation}

These closed-form phase shifts provide an intuitive physical interpretation: each PS phase shift corresponds to the steering phase at the equivalent center of its antenna group. As the group moves along the array, this center shifts accordingly, and hence the required PS phase shift changes linearly with the group index $l$. The two cases in \eqref{eq_phi_case1} and \eqref{eq_phi_case2} differ only because their antenna groups have different centers. Therefore, the proposed sparse-PS design approximates the ideal full-PS steering phases by assigning one representative phase to each antenna group, and the resulting beamforming loss comes from the residual phase mismatch inside each group.

\subsubsection{Case 2} If the phase of the $\bar{L}$-th entry of the steering vector, where $\bar{L}\le{N}$ and $\bar{L}\le{L+1}$, is equal to $2\pi$ or $-2\pi$, i.e.,\looseness=-1
\begin{equation}
	\big|(\bar{L}-1)\pi\sin(\theta)\big|=2\pi~\Rightarrow~|\sin(\theta)|=\frac{2}{\bar{L}-1}.
\end{equation}
Under this assumption, certain antenna elements may share identical steering phases. Specifically, phase overlaps occur when the corresponding antenna indices $n$ satisfy
\vspace{0pt}
\begin{equation}\label{eq_SRSU_case2_antenna_index}
	n=\bar{l}+\zeta\big(\bar{L}-1\big),~\bar{l}\in\mathcal{\bar{L}}\triangleq\big\{1,\ldots,\bar{L}-1\big\},
\end{equation}
where $\zeta\in\mathbb{N}$ should satisfy $n\le{N}$, i.e., $\zeta=0,\ldots,$ $\big\lfloor\big(N-\bar{l}\big)/\big(\bar{L}-1\big)\big\rfloor$. Accordingly, the steering phases contain only $\big(\bar L-1\big)$ distinct values, and antennas with identical steering phases can therefore be connected to the same PS without introducing intra-group phase mismatch. This optimally partitions the antennas into $\big(\bar L-1\big)$ groups, with the antenna indices in group $\mathcal{N}_{\bar{l}}$ given by \eqref{eq_SRSU_case2_antenna_index}. We connect the antennas in $\mathcal N_{\bar l}$ to the $\bar l$-th PS, for $\bar l=1,\ldots,\bar L-1$. If $\bar L-1<L$, the remaining PSs can be assigned by splitting existing antenna groups with internally identical steering phases and imposing the same phase shift on the resulting subgroups, which preserves the coherent sum and thus the full array gain. The optimal phase shift for the $\bar{l}$-th active PS connected to $\mathcal{N}_{\bar{l}}$ is given by
\begin{align}
	\bar{\phi}_{\bar{l}}&=\frac{\sum_{\zeta=0}^{\left\lfloor\frac{N-\bar{l}}{\bar{L}-1}\right\rfloor}\exp\big(\!-\!j\big(\bar{l}+\zeta\big(\bar{L}-1\big)-1\big)\pi\sin(\theta)\big)}{\bigg|\sum_{\zeta=0}^{\left\lfloor\frac{N-\bar{l}}{\bar{L}-1}\right\rfloor}\exp\big(\!-\!j\big(\bar{l}+\zeta\big(\bar{L}-1\big)-1\big)\pi\sin(\theta)\big)\bigg|}\nonumber\\[-3pt]
	&\overset{\mathrm{(d)}}{=}\exp\bigg(\!\!-\!j\frac{\big(\bar{l}-1\big)2\pi}{\bar{L}-1}\bigg)\label{eq_SRSU_case2_phi_d},~\text{or}\\[-2pt]
	&\overset{\mathrm{(e)}}{=}\exp\bigg(j\frac{\big(\bar{l}-1\big)2\pi}{\bar{L}-1}\bigg),~\bar{l}\in\mathcal{\bar{L}},\label{eq_SRSU_case2_phi_e}
\end{align}
where (d) holds for $\sin(\theta)=2/\big(\bar{L}-1\big)$, and (e) holds for $\sin(\theta)=-2/\big(\bar{L}-1\big)$.

\vspace{-1pt}
\begin{lemma}\label{Lemma1}
	\textit{Based on the antenna grouping strategy described in Case 2 and the corresponding optimal phase-shift expressions in \eqref{eq_SRSU_case2_phi_d} and \eqref{eq_SRSU_case2_phi_e}, \eqref{eq_SRSU_gain} can be further expressed as}
	\begin{align}
		\big|\mathbf{a}(\theta)^H\mathbf{\bar{C}}\bm{\bar{\phi}}\big|&=\sum\nolimits_{\bar{l}\in\bar{\mathcal{L}}}\bigg|\sum\nolimits_{\zeta=0}^{\left\lfloor\frac{N-\bar{l}}{\bar{L}-1}\right\rfloor}e^{\pm{j}\big(\big(\bar{l}-1\big)2\pi/\big(\bar{L}-1\big)+2\zeta\pi\big)}\bigg|\nonumber\\[-2pt]
		&=\sum\nolimits_{\bar{l}\in\bar{\mathcal{L}}}\bigg(\bigg\lfloor\frac{N-\bar{l}}{\bar{L}-1}\bigg\rfloor+1\bigg)=N.\label{eq_lemma1}
	\end{align}
\end{lemma}
\vspace{-8pt}
\begin{proof}
	Please refer to Appendix \ref{appendix_A}.
\end{proof}

\vspace{-6pt}
According to \textit{Lemma \ref{Lemma1}}, the antenna grouping strategy described in \textit{Case 2}, together with the phase shifts $\big\{\bar{\phi}_{\bar{l}}\big\}_{\bar{l}\in\mathcal{\bar{L}}}$, achieves the maximum array gain $\big|\mathbf{a}(\theta)^H\mathbf{\bar{C}}\bm{\bar{\phi}}\big|=N$, which indicates that reducing the PS count does not incur any beamforming performance loss.

\subsubsection{General Case} If the assumptions adopted in \textit{Case 1} and \textit{Case 2} do not hold, e.g., when the antenna count $N$ or the magnitude $|\sin(\theta)|$ becomes larger, a closed-form characterization is usually unavailable. In such general scenarios, we present the following proposition to reveal the relationship between the PS count $L$ and the maximum achievable array gain.\looseness=-1

\vspace{-1pt}
\begin{proposition}\label{Proposition1}
	\textit{The maximum array gain $\big|\mathbf{a}(\theta)^H\mathbf{\bar{C}}_\mathrm{opt}\bm{\bar{\phi}}_\mathrm{opt}\big|$ is monotonically nondecreasing w.r.t. the PS count $L$.}
\end{proposition}
\vspace{-9pt}
\begin{proof}
	Please refer to Appendix \ref{appendix_B}.
\end{proof}
\vspace{-4pt}
\textit{Proposition \ref{Proposition1}} indicates that one may choose the minimum PS count that achieves an acceptable array-gain loss relative to the full-PS architecture, rather than deploying excessive PSs.

\vspace{-6pt}
\subsection{Beam Switching with Common PS Connection}\label{subsec3_2}

\begin{algorithm}[t]
	\renewcommand{\algorithmicrequire}{\textbf{Input:}}
	\renewcommand{\algorithmicensure}{\textbf{Output:}}
	\caption{\hspace{2pt}MICP Algorithm for Solving Problem (P2.3)}
	\begin{algorithmic}[1]
		\label{alg1}
		\REQUIRE Initial point $\mathbf{\bar{C}}^{[0]}$ and convergence tolerance $\epsilon$.
		\ENSURE The optimized solution $\big\{\mathbf{\bar{C}}^\star,\big\{\bar{P}_k^\star\big\}_{k\in\mathcal{K}}\big\}$.
		\STATE Set iteration index $j=0$.
		\REPEAT
		\STATE Update $\big\{\mathbf{\bar{C}}^{[j+1]},\big\{\bar{P}_k^{[j+1]}\big\}_{k\in\mathcal{K}}\big\}$ by solving problem (P2.3) at a given point $\mathbf{\bar{C}}^{[j]}$.
		\STATE Set $j=j+1$.
		\UNTIL The fractional decrease in the objective value of problem (P2.3) is below $\epsilon$.
	\end{algorithmic}
\end{algorithm}
\setlength{\textfloatsep}{5pt}

In the beam-switching scenario (i.e., $K>1$), a common connection matrix supports multiple users/directions with distinct AoDs via a time-division manner. The phase-shift vectors $\left\{\bm{\bar{\phi}}_k\right\}_{k\in\mathcal{K}}$ can still be designed according to \eqref{eq_SRSU_phi_l_1} so as to align the phases of the entries in $\mathbf{h}_k^H\mathbf{\bar{C}}$. Specifically, we have\looseness=-1
\vspace{0pt}
\begin{equation}\label{eq_SRMU_phi_l}
	\bar{\phi}_{k,l}=\frac{\mathbf{\bar{C}}[:,l]^T\mathbf{a}(\theta_k)}{\big|\mathbf{\bar{C}}[:,l]^T\mathbf{a}(\theta_k)\big|},~k\in\mathcal{K},~l\in\mathcal{L}.
\end{equation}
By substituting the optimal phase-shift solution derived above into problem (P2), we can equivalently reformulate it as
\begin{subequations}
	\begin{alignat}{2}
		\mathrm{(P2.2)}:~&\underset{\mathbf{\bar{C}},\left\{\bar{P}_k\right\}_{k\in\mathcal{K}}}{\min}~&&\sum\nolimits_{k\in\mathcal{K}}\bar{P}_k\\[-4pt]
		&~\quad\mathrm{s.t.}&&\hspace{-2pt}\frac{\bar{P}_k\alpha_k}{N\sigma_k^2}\big\|\mathbf{a}(\theta_k)^H\mathbf{\bar{C}}\big\|_1^2\ge\gamma_k,~k\in\mathcal{K},\label{eq_P2_2_b}\\[-1pt]
		&&&\hspace{-1pt}\eqref{eq_P2_e},~\eqref{eq_P2_f},~\eqref{eq_P2_c}.\nonumber
	\end{alignat}
\end{subequations}
However, problem (P2.2) is still an intractable MINLP problem due to the non-convex QoS constraint \eqref{eq_P2_2_b} and the binary nature of $\mathbf{\bar{C}}$. To facilitate subsequent development, we equivalently rewrite \eqref{eq_P2_2_b} as
\vspace{-3pt}
\begin{equation}\label{eq_P2_2_b_re}
	\big\|\mathbf{a}(\theta_k)^H\mathbf{\bar{C}}\big\|_1\ge\sqrt{\frac{N\sigma_k^2\gamma_k}{\bar{P}_k\alpha_k}},~k\in\mathcal{K}.
\end{equation}
Let $f_k\big(\mathbf{\bar{C}}\big)\triangleq\big\|\mathbf{a}(\theta_k)^H\mathbf{\bar{C}}\big\|_1$. Thanks to the convexity of the left-hand side of \eqref{eq_P2_2_b_re}, we can derive a global lower bound of $f_k\big(\mathbf{\bar{C}}\big)$ as a surrogate function by employing its first-order Taylor-expansion-based affine under-estimator, yielding\footnote{Although the entries of the PS connection matrix $\mathbf{\bar{C}}$ are inherently binary, a valid lower bound of $f_k\big(\mathbf{\bar{C}}\big)$ w.r.t. $\mathbf{\bar{C}}$ can still be derived from a global affine lower bound of $f_k(\mathbf{\mathbf{A}})$ w.r.t. a continuous variable $\mathbf{A}$.}
\begin{equation}\label{eq_F2_LB}
	f_k\big(\mathbf{\bar{C}}\big)\ge f_k\big(\mathbf{\bar{C}}^{[j]}\big)+\operatorname{tr}\big(\big(\nabla_{\mathbf{\bar{C}}}f_k\big(\mathbf{\bar{C}}^{[j]}\big)\big)^T\big(\mathbf{\bar{C}}-\mathbf{\bar{C}}^{[j]}\big)\big),
\end{equation}
where $\mathbf{\bar{C}}^{[j]}$ denotes the given local point, and $\nabla_{\mathbf{\bar{C}}}f_k\big(\mathbf{\bar{C}}\big)$ is the gradient matrix over $\mathbf{\bar{C}}$, with its derivation provided in \textit{Appendix \ref{appendix_C}}. Accordingly, a convex subset of \eqref{eq_P2_2_b_re} is given by
\begin{equation}\label{eq_P2_3_b}
	f_k\big(\mathbf{\bar{C}}^{[j]}\big)\!+\!\operatorname{tr}\big(\!\big(\nabla_{\mathbf{\bar{C}}}f_k\big(\mathbf{\bar{C}}^{[j]}\big)\!\big)^T\!\big(\mathbf{\bar{C}}\!-\!\mathbf{\bar{C}}^{[j]}\big)\!\big)\!\ge\!\sqrt{\frac{N\sigma_k^2\gamma_k}{\bar{P}_k\alpha_k}},
\end{equation}
such that by replacing \eqref{eq_P2_2_b} in problem (P2.2) with \eqref{eq_P2_3_b}, an upper bound of the objective value can be obtained by solving\looseness=-1
\newcounter{al1}
\setcounter{al1}{\value{equation}}
\begin{subequations}
	\begin{alignat}{2}
		\mathrm{(P2.3)}:\quad&\underset{\mathbf{\bar{C}},\{\bar{P}_k\}_{k\in\mathcal{K}}}{\min}\quad&&\sum\nolimits_{k\in\mathcal{K}}\bar{P}_k\nonumber\\[-2pt]
		&\quad~\mathrm{s.t.}&&\eqref{eq_P2_e},\eqref{eq_P2_f},\eqref{eq_P2_c},\eqref{eq_P2_3_b}.\nonumber
	\end{alignat}
\end{subequations}

\noindent Problem (P2.3) is a mixed-integer convex program (MICP) and can be solved optimally by standard solvers such as CVX with MOSEK \cite{Grant_CVX_software}. The proposed optimization procedure is outlined in \textit{Algorithm \ref{alg1}}. The resulting objective sequence is monotonically nonincreasing since the previous solution remains feasible for the current update. Given that it is lower bounded by zero, this sequence is guaranteed to converge.
\setcounter{equation}{\value{al1}}

\vspace{-6pt}
\section{Multiple RF-Chains}\label{sec4}
In this section, we address the general multi-RF-chain scenario ($R\ge2$). Under such a configuration, problem (P1) is inherently difficult to solve directly because the binary PS connection and assignment matrices are deeply coupled with the continuous digital and analog beamformers. To obtain an efficient yet high-quality solution, we develop an iterative QoS-MM algorithm. The proposed algorithm alternately performs three main operations: i) a low-complexity digital-precoder update based on uplink--downlink duality; ii) a unit-modulus-preserving phase-shift MM update; and iii) a joint PS connection and assignment update built from linear MM surrogates. Below, the reformulated optimization problem is first presented explicitly, and then solved in three blocks.

\vspace{-10pt}
\subsection{Equivalent Signal Representation}
To facilitate subsequent optimization, we redefine the compact connection matrix $\mathbf{\hat{C}}$ as
\begin{equation}
	\mathbf{\hat{C}}\triangleq\big[\mathbf C_1^T,\mathbf C_2^T,\ldots,\mathbf C_R^T\big]\in\{0,1\}^{L\times N}.
\end{equation}
For user $k\in\mathcal{K}$, redefine the block-diagonal sparsified channel matrix $\mathbf{\hat{H}}_k$ as
{\small
\begin{align}
	\mathbf{\hat{H}}_k\!\triangleq\!\left[\begin{matrix}
	\mathbf h_k^*[1\!:\!M] \!\!\!&\!\!\! \mathbf 0^{M\times1} \!\!\!&\!\!\! \cdots \!\!\!&\!\!\! \mathbf 0^{M\times1} \\
	\mathbf 0^{M\times1} \!\!\!&\!\!\! \mathbf h_k^*[M\!\!+\!\!1\!:\!2M] \!\!\!&\!\!\! \cdots \!\!\!&\!\!\! \mathbf 0^{M\times1} \\[-7pt]
	\vdots \!\!\!&\!\!\! \vdots \!\!\!&\!\!\! \ddots \!\!\!&\!\!\! \vdots \\[-1pt]
	\mathbf 0^{M\times1} \!\!\!&\!\!\! \mathbf 0^{M\times1} \!\!\!&\!\!\! \cdots \!\!\!&\!\!\! \mathbf h_k^*[RM\!\!-\!\!M\!\!+\!\!1\!:\!RM]\end{matrix}\right]\!\in\mathbb C^{N\times R}\!.
\end{align}}%
Then, for stream $i$ at user $k$, the equivalent signal coefficient is\looseness=-1
\begin{equation}\label{eq_gki_def_new}
	g_{k,i}\triangleq\bm\phi^T\mathbf{\hat{C}}\mathbf{\hat{H}}_k\mathbf d_i,~k,i\in\mathcal{K}.
\end{equation}
Accordingly, the original problem (P1) can be rewritten as
\begin{subequations}\label{eq_P3_new}
	\begin{alignat}{2}
		\mathrm{(P3)}:~
		&\underset{\{\mathbf{d}_k\}_{k\in\mathcal{K}},\mathbf{\hat{C}},\mathbf S,\bm\phi}{\min}~
		&&\sum\nolimits_{k\in\mathcal{K}}\|\mathbf{d}_k\|_2^2\label{eq_P3_new_a}\\[-4pt]
		&\hspace{4pt}\mathrm{s.t.}
		&&\hspace{-31pt}\frac{|g_{k,k}|^2}{\sum\nolimits_{i\in\mathcal K,i\ne k}|g_{k,i}|^2\!+\!M\sigma_k^2}\ge\gamma_k,~k\in\mathcal K,\label{eq_P3_new_b}\\[-1pt]
		&&&\hspace{-31pt}\sum\nolimits_{l\in\mathcal L}\mathbf{\hat{C}}[l,n]=1,~n\in\mathcal N,\label{eq_P3_new_c}\\[-1pt]
		&&&\hspace{-31pt}\sum\nolimits_{r\in\mathcal R}\mathbf S[r,l]=1,~l\in\mathcal L,\label{eq_P3_new_d}\\[-1pt]
		&&&\hspace{-30pt}\mathbf{\hat{C}}[l,n]\le \mathbf S[r(n),l],~l\in\mathcal L,~n\in\mathcal N,\label{eq_P3_new_e}\\[-1pt]
		&&&\hspace{-30pt}\mathbf S[r,l]\le \sum\nolimits_{n\in\mathcal N_r}\!\!\!\mathbf{\hat{C}}[l,n],~r\in\mathcal R,~l\in\mathcal L,\label{eq_P3_new_f}\\[-1pt]
		&&&\hspace{-30pt}\mathbf{\hat{C}}[l,n]\in\{0,1\},~l\in\mathcal{L},~n\in\mathcal{N},\label{eq_P3_new_g}\\[-1pt]
		&&&\hspace{-30pt}\eqref{eq_P1_d},\eqref{eq_P1_e},\nonumber
	\end{alignat}
\end{subequations}
where $r(n)\triangleq\lfloor(n-1)/M\rfloor+1$ denotes the RF-chain/sub-array associated with the $n$-th antenna, and $\mathcal N_r\!\triangleq\{(r-1)M+1,\ldots,rM\}$ denotes the antenna set of the $r$-th sub-array. Constraints \eqref{eq_P3_new_c}--\eqref{eq_P3_new_f} are the equivalent compact form of \eqref{eq_PS_allocation_1}--\eqref{eq_PS_allocation_4}. In the following subsections, we derive the three update blocks of the proposed QoS-MM algorithm.

\vspace{-10pt}
\subsection{Digital Precoder Update}\label{subsec_42}
For fixed $\big\{\mathbf{\hat{C}},\mathbf S,\bm\phi\big\}$, define the effective channel vector $\mathbf g_k\in\mathbb C^{R\times1}$ as
\vspace{-2pt}
\begin{equation}
	\mathbf g_k^H\triangleq\bm\phi^T\mathbf{\hat{C}}\mathbf{\hat{H}}_k,~k\in\mathcal K.
	\label{eq_effective_channel_new}
\end{equation}
Then, $g_{k,i}=\mathbf g_k^H\mathbf d_i$. The digital precoder $\{\mathbf{d}_k\}_{k\in\mathcal{K}}$ update is
\begin{subequations}\label{eq_D_update_problem_new}
	\begin{alignat}{2}
		\mathrm{(P3.1)}:~&\underset{\{\mathbf d_k\}_{k\in\mathcal K}}{\min}~
		&&\sum\nolimits_{k\in\mathcal K}\|\mathbf d_k\|_2^2\label{eq_D_update_problem_new_a}\\[-5pt]
		&\hspace{-12pt}\mathrm{s.t.}\quad
		&&\hspace{-25pt}\frac{\big|\mathbf g_k^H\mathbf d_k\big|^2}{\sum\nolimits_{i\in\mathcal K,i\ne k}\big|\mathbf g_k^H\mathbf d_i\big|^2+M\sigma_k^2}\ge\gamma_k,~k\in\mathcal K.\label{eq_D_update_problem_new_b}
	\end{alignat}
\end{subequations}
Although problem (P3.1) can be solved as a second-order cone program (SOCP) \cite{WangHH_interference_MA}, we exploit uplink--downlink duality to obtain an exact low-complexity solution \cite{Schubert_up_downlink_duality}.

Let $n_k\triangleq M\sigma_k^2$. The QoS requirement for user $k$ in \eqref{eq_D_update_problem_new_b} can be equivalently written as
\vspace{0pt}
\begin{equation}
	r_k(\mathbf D)\triangleq\gamma_k\sum\nolimits_{i\in\mathcal K,i\ne k}\big|\mathbf g_k^H\mathbf d_i\big|^2\!-\big|\mathbf g_k^H\mathbf d_k\big|^2\!+\gamma_k n_k\le 0,
	\label{eq_D_residual_new}
\end{equation}
where $\mathbf{D}\triangleq[\mathbf{d}_1,\ldots,\mathbf{d}_K]$. Let $\eta_k\ge0$ be the Lagrange multiplier associated with \eqref{eq_D_residual_new}, and define the virtual uplink power\looseness=-1
\begin{equation}
	\chi_k\triangleq\eta_k\gamma_k,~k\in\mathcal{K}.
\end{equation}
The corresponding virtual uplink covariance matrix is
\begin{equation}
	\mathbf A(\bm\chi)\triangleq\mathbf I^R+
	\sum\nolimits_{k\in\mathcal K}\chi_k\mathbf g_k\mathbf g_k^H,
	\label{eq_A_chi_new}
\end{equation}
where $\bm{\chi}\triangleq[\chi_1,\ldots,\chi_K]^T$. For the dual function to be bounded below, we need
\begin{equation}\label{eq_dual_ineq}
	\mathbf{A}(\bm{\chi})-\chi_k\left(1+\gamma_k^{-1}\right)\mathbf{g}_k\mathbf{g}_k^H\succeq\mathbf{0},~k\in\mathcal{K}.
\end{equation}
Using the rank-one matrix inequality, \eqref{eq_dual_ineq} is equivalent to
\begin{equation}\label{eq_dual_ineq_re}
	\chi_k\left(1+\gamma_k^{-1}\right)\mathbf{g}_k^H\mathbf{A}(\bm{\chi})^{-1}\mathbf{g}_k\le1,~k\in\mathcal{K}.
\end{equation}
The above inequalities are tight at the optimum, yielding the following fixed-point iteration given the current point $\bm{\chi}^{[i]}$:
\begin{equation}\label{eq_chi_fixed_point_new}
	\chi_k^{[i+1]}=\frac{\gamma_k}{(1+\gamma_k)\mathbf g_k^H\mathbf A(\bm{\chi}^{[i]})^{-1}\mathbf g_k},~k\in\mathcal K.
\end{equation}

From the Karush-Kuhn-Tucker (KKT) stationarity condition, it holds that
\begin{equation}\label{eq_KKT_stationnarity}
	\left(\mathbf{A}(\bm{\chi})-\chi_k\left(1+\gamma_k^{-1}\right)\mathbf{g}_k\mathbf{g}_k^H\right)\mathbf{d}_k=\mathbf{0},~k\in\mathcal{K}.
\end{equation}
Since $\mathbf{A}(\bm{\chi})\succ\mathbf{0}$, \eqref{eq_KKT_stationnarity} implies that $\mathbf{d}_k$ is aligned with $\mathbf{A}(\bm{\chi})^{-1}\mathbf{g}_k$. Thus, for the converged $\bm{\chi}$, the normalized optimal downlink beam direction is expressed as
\begin{equation}\label{eq_d_direction_new}
	\mathbf{d}_k^\mathrm{N}=\frac{\mathbf A(\bm\chi)^{-1}\mathbf g_k}{\|\mathbf A(\bm\chi)^{-1}\mathbf g_k\|_2},~k\in\mathcal K.
\end{equation}

The downlink transmit powers $\{p_k\}_{k\in\mathcal{K}}$ are then derived by enforcing all QoS constraints to be active. Specifically, by defining $\mathbf B\in\mathbb R^{K\times K}$ as
\begin{equation}
	\mathbf B[k,k]\triangleq\frac{\big|\mathbf g_k^H\mathbf d_k^\mathrm{N}\big|^2}{\gamma_k},\quad\mathbf B[k,i]=-\big|\mathbf g_k^H\mathbf d_i^\mathrm{N}\big|^2,~\forall i\ne k,
\end{equation}
the power vector $\mathbf{p}\triangleq[p_1,\ldots,p_K]^T$ is determined by solving
\begin{equation}\label{eq_power_linear_new}
	\mathbf B\mathbf{p}=[n_1,\ldots,n_K]^T.
\end{equation}
Consequently, the final digital precoders $\{\mathbf{d}_k\}_{k\in\mathcal{K}}$ are given by\looseness=-1
\begin{equation}\label{eq_d_final_new}
	\mathbf d_k=\sqrt{p_k}\,\mathbf d_k^\mathrm{N} \exp\big(\!-\!j\arg\big(\mathbf g_k^H\mathbf d_k^\mathrm{N}\big)\big),~k\in\mathcal K.
\end{equation}
Furthermore, the QoS multipliers adopted in the subsequent MM updates are recovered via
\begin{equation}\label{eq_eta_from_chi_new}
	\eta_k=\chi_k/\gamma_k,~k\in\mathcal K.
\end{equation}
Upon convergence of the fixed-point iteration, this update yields the optimal solution to problem (P3.1).

\vspace{-8pt}
\subsection{PS Phase-Shift Update}\label{subsec_43}
For fixed $\big\{\!\big\{\!\mathbf{d}_k\!\big\}_{k\in\mathcal{K}},\!\mathbf{\hat{C}},\mathbf S\big\}$, the transmit power $\sum_{k\in\mathcal{K}}\!\|\mathbf{d}_k\|_2^2$ is independent of $\bm\phi$. Hence, we update the PS phase shifts $\bm\phi$ by minimizing dual-weighted QoS residuals, which measures the difficulty of satisfying the QoS requirements under the current digital precoders and PS connections. After obtaining update PS phase shifts, the digital precoders are re-optimized to evaluate the actual transmit power. To this end, we define
\begin{equation}\label{eq_qki_v_def_new}
	\mathbf q_{k,i}\triangleq\mathbf{\hat{C}}\mathbf{\hat{H}}_k\mathbf d_i\in\mathbb{C}^{L\times1},\quad
	\mathbf v\triangleq\bm\phi^*.
\end{equation}
Since $\mathbf v^H=\bm\phi^T$, it holds that
\begin{equation}
	g_{k,i}=\bm\phi^T\mathbf q_{k,i}=\mathbf v^H\mathbf q_{k,i}.
\end{equation}
The QoS residual for user $k$ w.r.t. $\mathbf v$ is expressed as
\begin{equation}\label{eq_phi_residual_new}
	r_k(\mathbf v)\triangleq
	\gamma_k\sum\nolimits_{i\in\mathcal K,i\ne k}\!\big|\mathbf v^H\mathbf q_{k,i}\big|^2
	\!-\big|\mathbf v^H\mathbf q_{k,k}\big|^2\!+
	\gamma_k M\sigma_k^2,
\end{equation}
where the three terms respectively represent the aggregate interference, the desired-signal power, and the noise-related component for user $k$. Thus, under fixed digital precoders and PS connections, decreasing $r_k(\mathbf v)$ effectively enlarges the QoS margin for user $k$. However, because the QoS tightness varies across users/directions, simply reducing these residuals uniformly is inappropriate. Instead, an aggregate weighted QoS residual is formulated by leveraging the dual weights $\{\eta_k\}_{k\in\mathcal{K}}$ acquired from the preceding digital-precoder update, yielding\looseness=-1
\begin{align}
	\!J_{\phi}(\mathbf v)\!\triangleq\!\sum\nolimits_{k\in\mathcal K}\hspace{-2pt}\eta_k r_k(\mathbf v)\!=\!\mathbf v^H\mathbf Q_{\phi}\mathbf v\!+\!
	\sum\nolimits_{k\in\mathcal K}\hspace{-2pt}\eta_k\gamma_kM\sigma_k^2,
	\label{eq_J_phi_new}
\end{align}
where $\mathbf Q_{\phi}\triangleq\sum\nolimits_{k\in\mathcal K}\eta_k\big(\gamma_k\sum\nolimits_{i\in\mathcal K,i\ne k}\!\mathbf q_{k,i}\mathbf q_{k,i}^H-\mathbf q_{k,k}\mathbf q_{k,k}^H\big)$.

The matrix $\mathbf Q_{\phi}$ is generally indefinite, which renders the following weighted QoS residual minimization problem non-convex and highly intractable:
\begin{subequations}
	\begin{alignat}{2}
		\mathrm{(P3.2)}:\quad&\underset{\mathbf{v}}{\min}\quad
		&&\mathbf{v}^H\mathbf{Q}_\phi\mathbf{v}\label{eq_P3_2_a}\\[-4pt]
		&~\mathrm{s.t.}
		&&|v_l|=1,~l\in\mathcal{L},\label{eq_P3_2_b}
	\end{alignat}
\end{subequations}
where $\mathbf{v}\triangleq[v_1,\ldots,v_L]^T$ represents the conjugate unit-modulus phase-shift vector. To tackle problem (P3.2), a surrogate function is constructed based on the MM algorithmic framework \cite{SunY_surrogate_function}. Choose $\tau_{\phi}\ge\lambda_{\max}(\mathbf Q_{\phi})$, such that $\mathbf{Z}_\phi\triangleq\tau_{\phi}\mathbf I^L-\mathbf Q_{\phi}\succeq\mathbf 0$. To avoid computationally intensive eigenvalue decomposition, a low-complexity choice for $\tau_{\phi}$ is the Gershgorin bound \cite{Horn_matrix_2012}:
\begin{equation}
	\tau_{\phi}=\max_{l\in\mathcal{L}}\,\sum\nolimits_{l'\in\mathcal{L}}\big|\mathbf Q_{\phi}[l,l']\big|.
	\label{eq_tau_phi_new}
\end{equation}
For any feasible $\mathbf{v}$ and given the local point $\mathbf{v}^{[j]}$, it holds that
\begin{equation}
	\big(\mathbf{v}-\mathbf{v}^{[j]}\big)^H\mathbf{Z}_\phi\big(\mathbf{v}-\mathbf{v}^{[j]}\big)\ge0.
\end{equation}
Due to the unit-modulus constraints $\eqref{eq_P3_2_b}$, we have $\mathbf{v}^H\mathbf{v}=L$. Hence, the objective function in (P3.2) can be majorized as
\begin{align}
	\mathbf{v}^H\mathbf{Q}_\phi\mathbf{v}&=\tau_{\phi}L-\mathbf{v}^H\mathbf{Z}_\phi\mathbf{v}\\[-3pt]
	&\le\tau_{\phi}L-2\Re\big\{\mathbf{v}^H\mathbf{Z}_\phi\mathbf{v}^{[j]}\big\}+\big(\mathbf{v}^{[j]}\big)^H\mathbf{Z}_\phi\mathbf{v}^{[j]}.
\end{align}
This upper bound is globally tight at $\mathbf{v}=\mathbf{v}^{[j]}$. Minimizing this surrogate function subject to \eqref{eq_P3_2_b} is equivalent to maximizing $\Re\big\{\mathbf{v}^H\mathbf{Z}_\phi\mathbf{v}^{[j]}\big\}$, which yields the following closed-form optimal solution:
\vspace{-3pt}
\begin{equation}\label{eq_v_candidate_new}
	\mathbf v=
	\exp\big(\,j\arg\big(\big(\tau_{\phi}\mathbf I^L-\mathbf Q_{\phi}\big)\mathbf v^{[j]}\big)\big),
\end{equation}
where the exponential and argument operations are applied element-wise. The corresponding optimal phase-shift vector is then recovered as $\bm\phi=\mathbf v^*$. Upon updating the PS phase shifts $\bm\phi$, the digital precoders are re-optimized according to subsection \ref{subsec_42}. This phase-shift update is accepted only if it leads to a reduction in the actual transmit power.


\vspace{-4pt}
\subsection{PS Connection and Assignment Update}\label{subsec_dynamic_connection_new}
For fixed $\{\{\mathbf{d}_k\}_{k\in\mathcal{K}},\bm\phi\}$, the transmit-power objective does not depend on the binary variables $\big\{\mathbf{\hat{C}},\mathbf{S}\big\}$. Therefore, the PS connection and assignment matrices can also be updated by reducing the dual-weighted QoS residual, followed by re-optimizing the digital precoders and PS phase shifts to evaluate the actual transmit power. Let $\mathbf{\hat{c}}\triangleq\operatorname{vec}\big(\mathbf {\hat{C}}\big)\in\{0,1\}^{LN\times1}$.
For each user-stream pair $(k,i)$, define
\begin{equation}\label{eq_bki_def_new}
	\mathbf z_{k,i}\triangleq\mathbf{\hat{H}}_k\mathbf d_i,\quad
	\mathbf b_{k,i}\triangleq\operatorname{vec}\big(\bm\phi\mathbf z_{k,i}^{T}\big),~k,i\in\mathcal{K}.
\end{equation}
Then, the equivalent signal coefficient can be rewritten as
\begin{equation}\label{eq_g_b_c_new}
	g_{k,i}=
	\bm\phi^T\mathbf{\hat{C}}\mathbf{z}_{k,i}=
	\mathbf b_{k,i}^T\mathbf{\hat{c}},~k,i\in\mathcal{K}.
\end{equation}
The QoS constraint for user $k\in\mathcal{K}$ in \eqref{eq_P3_new_b} is equivalent to
\vspace{0pt}
\begin{equation}\label{eq_rk_c_expanded}
	r_k(\mathbf{\hat{c}})\triangleq
	\gamma_k\sum\nolimits_{i\in\mathcal K,i\ne k}\!\big|\mathbf b_{k,i}^T\mathbf{\hat{c}}\big|^2\!-
	\big|\mathbf b_{k,k}^T\mathbf{\hat{c}}\big|^2\!+
	\gamma_kM\sigma_k^2\le0,
\end{equation}
where $r_k(\mathbf{\hat{c}})$ is the QoS residual for user $k\in\mathcal{K}$. Since $\mathbf{\hat{c}}$ is real-valued and binary, we have
\begin{equation}
	\big|\mathbf b_{k,i}^T\mathbf{\hat{c}}\big|^2=
	\mathbf{\hat{c}}^T\Re\big\{\mathbf b_{k,i}^{*}\mathbf b_{k,i}^{T}\big\}\mathbf{\hat{c}},~k,i\in\mathcal{K}.
\end{equation}
Accordingly, \eqref{eq_rk_c_expanded} can be expressed as
\begin{equation}\label{eq_Qk_single_def_new}
	r_k(\mathbf{\hat{c}})=
	\mathbf{\hat{c}}^T\mathbf Q_k\mathbf{\hat{c}}+\gamma_kM\sigma_k^2,~k\in\mathcal{K},
\end{equation}
where $\mathbf Q_k\triangleq\gamma_k\sum\nolimits_{i\in\mathcal K,i\ne k}\Re\big\{\mathbf{b}_{k,i}^*\mathbf{b}_{k,i}^T\big\}-\Re\big\{\mathbf{b}_{k,k}^*\mathbf{b}_{k,k}^T\big\}$. By exploiting the QoS dual weights $\{\eta_k\}_{k\in\mathcal{K}}$ derived from the previous digital-precoder update, we form a dual-weighted aggregate residual to emphasize users whose QoS constraints are more critical to the transmit power, expressed as
\begin{equation}\label{eq_JC_new}
	\!\!\!J_C(\mathbf{\hat{c}})\triangleq
	\sum\nolimits_{k\in\mathcal K}\!\eta_k r_k(\mathbf{\hat{c}})=
	\mathbf{\hat{c}}^T\mathbf Q_C\mathbf{\hat{c}}+
	\sum\nolimits_{k\in\mathcal K}\!\eta_k\gamma_kM\sigma_k^2,
\end{equation}
where $\mathbf Q_C\triangleq\sum\nolimits_{k\in\mathcal K}\eta_k\mathbf Q_k$. The matrices $\mathbf Q_k$ and $\mathbf Q_C$ are indefinite because the desired-signal terms enter the QoS residuals with negative signs. Hence, directly minimizing \eqref{eq_JC_new} over binary variables $\big\{\mathbf{\hat{C}},\mathbf{S}\big\}$ results in a non-convex mixed-integer quadratic program. To obtain tractable PS connection and assignment updates, we next construct linear MM upper bounds for the individual QoS residuals and their weighted aggregate.\looseness=-1

\subsubsection{Linear MM Surrogate}
For each user $k$, choose $\beta_k\ge\lambda_{\max}(\mathbf{Q}_k)$ such that $\beta_k\mathbf I^{LN}-\mathbf Q_k\succeq\mathbf 0$. A closed-form choice is\looseness=-1
\begin{equation}\label{eq_beta_k_closed_form}
	\beta_k=
	\gamma_kL\sum\nolimits_{i\in\mathcal K,i\ne k}\big\|\mathbf{\hat{H}}_k\mathbf d_i\big\|_2^2,~k\in\mathcal{K},
\end{equation}
whose validity is proved in Appendix \ref{appendix_D}. \eqref{eq_beta_k_closed_form} requires neither explicit construction of the $LN\times LN$ matrix $\mathbf Q_k$ nor eigenvalue decomposition. At the current iteration point $\mathbf c^{[t]}$, define
\begin{align}
	&\bm\xi_k\triangleq
	2\mathbf Q_k\mathbf{\hat{c}}^{[t]}+
	\beta_k\big(\mathbf 1-2\mathbf{\hat{c}}^{[t]}\big),~k\in\mathcal{K},\label{eq_xi_k_MM}\\
	&\kappa_k\triangleq
	\big(\mathbf{\hat{c}}^{[t]}\big)^T\big(\beta_k\mathbf I^{LN}-\mathbf Q_k\big)\mathbf{\hat{c}}^{[t]}+
	\gamma_kM\sigma_k^2,~k\in\mathcal K.\label{eq_kappa_k_MM}
\end{align}
Then, a linear MM upper bound for the $k$-th QoS residual is
\begin{equation}\label{eq_residual_surrogate_new}
	r_k(\mathbf{\hat{c}})\le\hat r_k\big(\mathbf{\hat{c}}\mid\mathbf{\hat{c}}^{[t]}\big)\triangleq
	\bm\xi_k^T\mathbf{\hat{c}}+\kappa_k,~k\in\mathcal{K},
\end{equation}
where equality holds at $\mathbf{\hat{c}}=\mathbf{\hat{c}}^{[t]}$. By taking the dual-weighted sum of the individual upper bounds in \eqref{eq_residual_surrogate_new}, define
\begin{align}
	&\beta_C\triangleq\sum\nolimits_{k\in\mathcal{K}}\eta_k\beta_k,~\\
	&\bm\xi_C\triangleq\sum\nolimits_{k\in\mathcal K}\eta_k\bm\xi_k=2\mathbf Q_C\mathbf{\hat{c}}^{[t]}+\beta_C\big(\mathbf 1^{LN}\!\!-2\mathbf{\hat{c}}^{[t]}\big),\label{eq_xi_C_MM}\\
	&\kappa_C\!\triangleq\!\sum_{k\in\mathcal K}\!\eta_k\kappa_k\!=\!\big(\mathbf{\hat{c}}^{[t]}\big)\!^T\!\big(\beta_C\mathbf I^{LN}\!\!\!-\!\mathbf Q_C\big)\mathbf{\hat{c}}^{[t]}\!+\!\!\sum_{k\in\mathcal K}\!\eta_k\gamma_kn_k.
\end{align}
The dual-weighted aggregate linear MM surrogate is therefore
\begin{equation}\label{eq_JC_surrogate_new}
	J_C(\mathbf{\hat{c}})\le\hat J_C\big(\mathbf{\hat{c}}\mid\mathbf{\hat{c}}^{[t]}\big)\triangleq
	\bm\xi_C^T\,\mathbf{\hat{c}}+\kappa_C,
\end{equation}
where equality holds at $\mathbf{\hat{c}}=\mathbf{\hat{c}}^{[t]}$. Furthermore, since $\eta_k\ge0$ and $\beta_k\mathbf I^{LN}-\mathbf Q_k\succeq\mathbf 0$, $\forall\,k\in\mathcal K$, it holds that
\begin{equation}\label{eq_beta_C_psd_MM}
	\beta_C\mathbf I^{LN}-\mathbf Q_C=
	\sum\nolimits_{k\in\mathcal K}\eta_k\big(\beta_k\mathbf I^{LN}-\mathbf Q_k\big)\succeq\mathbf 0.
\end{equation}
The derivations of linear MM surrogates $\hat r_k\big(\mathbf{\hat{c}}\!\mid\!\mathbf{\hat{c}}^{[t]}\big)$ and $\hat J_C\big(\mathbf{\hat{c}}\mid\mathbf{\hat{c}}^{[t]}\big)$ are provided in Appendix \ref{appendix_E}. In implementation, neither $\mathbf Q_k$ nor $\mathbf Q_C$ needs to be explicitly stored. Specifically, the term $\mathbf Q_k\mathbf{\hat{c}}^{[t]}$ in \eqref{eq_xi_k_MM} and \eqref{eq_xi_C_MM} can be computed as
\begin{align}\label{eq_matrix_free_gradient_MM}
	\mathbf Q_k\mathbf c^{[t]}=
	\gamma_k\sum\nolimits_{i\in\mathcal K, i\ne k}\!\Re\big\{\mathbf b_{k,i}^{*}g_{k,i}^{[t]}\big\}-
	\Re\big\{\mathbf b_{k,k}^{*}g_{k,k}^{[t]}\big\},
\end{align}
where $\!g_{k,i}^{[t]}\!\!=\!\mathbf b_{k,i}^T\mathbf{\hat{c}}^{[t]}\!$ is the current equivalent signal coefficients.

\subsubsection{Joint PS Connection and Assignment Update}
The binary variables $\big\{\mathbf{\hat{C}},\mathbf{S}\big\}$ are coupled: the connection matrix $\mathbf{\hat{C}}$ determines which PS connects to each antenna, while the assignment matrix $\mathbf{S}$ determines the RF-chain to which each PS belongs. Thus, rather than updating them separately, we jointly optimize $\big\{\mathbf{\hat{C}},\mathbf{S}\big\}$ under the dual-weighted aggregate linear MM surrogate $\hat J_C\big(\mathbf{\hat{c}}\mid\mathbf{\hat{c}}^{[t]}\big)$ in \eqref{eq_JC_surrogate_new}. Since $\kappa_C$ is independent of $\big\{\mathbf{\hat{C}},\mathbf{S}\big\}$, it can be omitted from the objective without compromising optimality. First, define the Hamming distance of $\mathbf{S}$ as\looseness=-1
\begin{equation}\label{eq_S_distance_new}
	d_S\triangleq\!\sum_{r\in\mathcal R}\sum_{l\in\mathcal L}\big[\mathbf S^{[t]}[r,l](1\!-\!\mathbf S[r,l])\!+\!\big(1\!-\!\mathbf S^{[t]}[r,l]\big)\mathbf S[r,l]\big].
\end{equation}
The joint PS connection and assignment optimization is
\begin{subequations}
	\begin{alignat}{2}
		\mathrm{(P3.3)}:~&\underset{\mathbf{\hat{C}},\mathbf S,\{s_k\}_{k\in\mathcal{K}}}{\min}&&\sum_{l\in\mathcal L}\!\sum_{n\in\mathcal N}\!\bm\xi_C[l,n]\mathbf{\hat{C}}[l,n]\!+\!\rho\!\sum_{k\in\mathcal K}\!s_k\!\label{eq_S_profile_MILP_new_a}\\
		&\hspace{-10pt}\mathrm{s.t.}&&\hspace{-31pt}\Delta_{\min}\le d_S\le\Delta_{\max},\label{eq_S_profile_MILP_new_b}\\
		&&&\hspace{-32pt}\sum_{l\in\mathcal L}\!\sum_{n\in\mathcal N}\!\bm\xi_k[l,n]\mathbf{\hat{C}}[l,n]\!+\!\kappa_k\le s_k,~k\in\mathcal K,\label{eq_S_profile_MILP_new_c}\\[-1pt]
		&&&\hspace{-31pt}s_k\ge0,~k\in\mathcal{K},\label{eq_S_profile_MILP_new_d}\\[-1pt]
		&&&\hspace{-31pt}\eqref{eq_P1_d},\eqref{eq_P3_new_c}-\eqref{eq_P3_new_g}.\nonumber
	\end{alignat}
\end{subequations}
The nonnegative slack variables $\{s_k\}_{k\in\mathcal K}$ in constraint \eqref{eq_S_profile_MILP_new_c} enable the exploration of PS connections whose individual QoS residual upper bounds $\hat{r}_k(\mathbf{\hat{c}}\mid\mathbf{\hat{c}}^{[t]})$ may become temporarily positive, while the penalty factor $\rho>0$ discourages such individual residual deterioration when reducing the overall dual-weighted aggregate surrogate $J_C\big(\mathbf{\hat{c}}\!\mid\!\mathbf{\hat{c}}^{[t]}\big)$. Moreover, $[\Delta_{\min},\Delta_{\max}]$ specifies a Hamming-distance range, which not only avoids an excessively large search space in each iteration but also prevents the PS assignment matrix $\mathbf{S}$ from getting stuck at local points. Noting that one reassigned PS changes two binary entries in $\mathbf S$, a Hamming distance of $2q$ corresponds to reassigning $q$ PSs among the RF-chains. Problem (P3.3) is an MILP and can be solved efficiently using CVX with the MOSEK solver. The optimized $\big\{\mathbf{\hat{C}},\mathbf{S}\big\}$ are subsequently evaluated by re-optimizing the continuous beamforming variables $\{\{\mathbf{d}_k\}_{k\in\mathcal{K}},\bm{\phi}\}$ based on subsections \ref{subsec_42} and \ref{subsec_43}, and are accepted only if the actual transmit power is reduced.

\begin{algorithm}[t]
	\renewcommand{\algorithmicrequire}{\textbf{Input:}}
	\renewcommand{\algorithmicensure}{\textbf{Output:}}
	\caption{\hspace{2pt}QoS-MM Algorithm for Solving Problem (P3)}
	\label{alg_profiledS_qosmm_new}
	\begin{algorithmic}[1]
			\REQUIRE Initial points $\big\{\mathbf{\hat{C}}^{[0]},\mathbf{S}^{[0]},\bm{\phi}^{[0]},\bm{\chi}^{[0]}\big\}$ and convergence tolerance $\epsilon$.
			\ENSURE The optimized solution $\big\{\{\mathbf{d}_k\}_{k\in\mathcal{K}}^\star,\mathbf{\hat{C}}^\star,\mathbf S^\star,\bm\phi^\star\big\}$.
			\STATE Set iteration index $x=0$.
			\REPEAT
			\STATE Update $\{\mathbf{d}_k^{[x+1]}\}_{k\in\mathcal{K}}$ by solving (P3.1) using closed-form solution \eqref{eq_d_final_new}, and compute $\big\{\eta_k^{[x+1]}\big\}_{k\in\mathcal{K}}$ by \eqref{eq_eta_from_chi_new}.\looseness=-1
			\STATE Update $\bm{\phi}=\mathbf{v}^*$ by solving (P3.2) using closed-form solution \eqref{eq_v_candidate_new}; re-optimize $\{\mathbf{d}_k^{[x+1]}\}_{k\in\mathcal{K}}$ by step 3 and accept $\bm{\phi}^{[x+1]}$ only if the actual transmit power decreases.
			\STATE Update $\{\mathbf{\hat{C}},\mathbf{S}\}$ by solving (P3.3); re-optimize $\big\{\bm\phi^{[x+1]},\big.$ $\big.\{\mathbf{d}_k^{[x+1]}\}_{k\in\mathcal{K}}\big\}$ by steps 4 and 3, and accept $\big\{\mathbf{\hat{C}}^{[x+1]},\big.$ $\big.\mathbf{S}^{[x+1]}\big\}$ only if the actual transmit power decreases.
			\STATE Set $x=x+1$.
			\UNTIL The fractional decrease in the actual transmit power is below $\epsilon$.
	\end{algorithmic}
\end{algorithm}
\setlength{\textfloatsep}{5pt}

\vspace{-9pt}
\subsection{Overall Algorithm and Convergence}
The proposed QoS-MM algorithm is summarized in \textit{Algorithm \ref{alg_profiledS_qosmm_new}}. After each update, the newly optimized PS phase shifts and connections are verified by re-optimizing the continuous beamforming variables under the QoS requirements, and are accepted only if the resulting transmit power decreases. Consequently, the accepted transmit-power sequence is monotonically nonincreasing and lower bounded by zero, thereby guaranteeing convergence of the objective sequence.

\vspace{-5pt}
\section{Numerical Results}\label{sec5}

\begin{table}[t]
	\centering
	\vspace{-5pt}
	\caption{Summary of Algorithm Complexity}
	\label{tab_complexity}
	\scriptsize
	\renewcommand{\arraystretch}{1}
	\vspace{-5pt}
	\begin{tabular}{
			@{}
			>{\centering\arraybackslash}m{1.1cm}
			@{\hspace{6pt}}
			>{\raggedright\arraybackslash}m{1.75cm}
			@{\hspace{6pt}}
			>{\raggedright\arraybackslash}m{2.9cm}
			@{\hspace{6pt}}
			>{\raggedright\arraybackslash}m{2.41cm}
			@{}
		}
		\toprule
		\textbf{Algorithm}
		& \textbf{Binary variables}
		& \textbf{Computational bottleneck}
		& \textbf{Iterations \& runtime} \\
		\midrule
		MICP
		& $NL$ in $\bar{\mathbf C}$
		& MICP in (P2.3); complexity is dominant by \cite{Polik_interior_point_book,WuYF_RIS_GBD}: $I_1\mathcal{O}\big(I_\mathrm{B}\!\log\!\frac{1}{\iota}(NL\!+\!K)^3\big)$.
		& $I_1=3.4~\mathrm{iterations}$; $0.36~\mathrm{s/iteration}$. \\
		\midrule
		QoS-MM
		& $NL$ in $\hat{\mathbf C}$ and $RL$ in $\mathbf S$
		& MILP in (P3.3); complexity is dominant by \cite{WuYF_RIS_GBD}: $I_2\mathcal{O}\big(I_\mathrm{B}(NL)^2\big)$.
		& $I_2=38.5~\mathrm{iterations}$; $10.82~\mathrm{s/iteration}$. \\
		\bottomrule
	\end{tabular}
\end{table}

This section presents numerical results to comprehensively evaluate the proposed sparse-PS static-connection architecture and the corresponding optimization algorithms under two representative communication scenarios: i) a single-RF-chain system and ii) a multi-RF-chain system. Unless otherwise specified, the LoS channel model in \eqref{eq_channel} is adopted. The AoDs $\theta_k$, $k\in\mathcal{K}$, of the candidate service users/directions are randomly generated within $[-\pi/2,\pi/2]$ over $100$ channel realizations. The BS-to-user distances, required SNR/SINR thresholds, noise powers, and channel power gains are set uniformly for all users as $d_k=d=50$~m, $\gamma_k=\gamma=10$~dB~$=10$, $\sigma_k^2=\sigma^2=-80$~dBm~$=10^{-11}$~W, and $\alpha_k=\beta_0 d^{-\xi_0}$, respectively, where $\beta_0=-62$~dB~$=10^{-6.2}$, denotes the reference channel power gain at $1$~m, and $\xi_0=2.1$ is the path-loss exponent \cite{ZhangYP_mmChannel,Rappaport_mmChannel}. The benchmark schemes are defined as follows:\looseness=-1

1) \textbf{Full-PS}: This benchmark considers the sub-connected architecture with the maximum possible PS count, i.e., $L=N$, where each antenna is configured with an independent PS. For each simulation scenario, the beamformers are designed with the correspondingly proposed algorithm tailored to that specific scenario, and this scenario-matched beamforming scheme is followed by all subsequent benchmarks. The full-PS architecture serves as an idealized reference for evaluating the proposed sparse-PS connection designs.

2) \textbf{Monte Carlo (MC) conn.}: The connection matrix is randomly generated under the PS connection constraints, and $r_\mathrm{a}$ independent random matrices are constructed. The best performance achieved among all these matrices is then reported.

3) \textbf{Cyclic conn.}: In this benchmark, antennas are connected to PSs in a cyclic interleaving manner. Specifically, the $n$-th ($n\in\mathcal{N}$) antenna is connected to PS~$(((n-1)\bmod L)+1)$. Next, by applying the same cyclic assignment pattern in each sub-array, this connection strategy is extended to the multi-RF-chain case with a total of $L$ PSs available.

4) \textbf{Adjacent conn.}: Adjacent antennas are connected to each PS sequentially in ascending index order. Let $X=\lfloor N/L\rfloor$ and $Y=N\!\!\mod L$. For $Y=0$, each PS is assigned $X$ consecutive antennas in sequence. For $Y>0$, the first $Y$ PSs take $(X+1)$ antennas each, and the rest take $X$ antennas each. This strategy is readily extended to the multi-RF-chain case, where the total number of available PSs is $L$.\looseness=-1

For reproducibility, the PS connection is initialized with a randomly generated binary matrix satisfying the hardware connection constraints. If needed, the initial PS assignment matrix $\mathbf{S}^{[0]}$ is constructed accordingly; the PS phase shifts $\bm{\phi}^{[0]}$ are independently initialized on the unit circle, and the initial virtual uplink powers are set to $\bm{\chi}^{[0]}=\mathbf{0}$. The convergence tolerance $\epsilon$, the uplink--downlink duality fixed-point tolerance, and the MOSEK solver tolerance are set to $10^{-4}$, $10^{-5}$, and the default, respectively. For each channel realization, the PS connection is optimized offline once for the corresponding candidate user/direction set, and then kept fixed during transmissions. Furthermore, the algorithm complexities are summarized in Table~\ref{tab_complexity}, where the simulation parameters are configured as $N=16$, $L=8$, $K=4$ for the MICP algorithm and $N=64$, $L=32$, $R=8$, $K=6$ for the QoS-MM algorithm. In the table, $\mathcal{O}(\cdot)$ denotes the Big-O complexity; $I_1$ and $I_2$ represent the average iterations of the MICP and QoS-MM algorithms, respectively; $I_\mathrm{B}$ is the number of branch-and-bound iterations; and $\iota>0$ is the convergence tolerance of the interior-point method. The algorithms are implemented in MATLAB R2022b and executed on a desktop computer with an Intel Core i9-14900KS CPU, utilizing MOSEK as the CVX solver.\looseness=-1

\vspace{-2pt}
\subsection{Single-RF-Chain System}
We first consider a single RF-chain system with $N=16$ antennas, where a static PS-to-antenna connection matrix is optimized to support beam switching among $K$ candidate users/directions via time-division multiplexing. The results obtained by \textit{Algorithm \ref{alg1}} proposed in Section \ref{subsec3_2} are referred to as \textit{Proposed MICP}, while the number of random connection matrices for \textit{MC conn.} is set to $r_\mathrm{a}=20$. 

\begin{figure}[t]
	\centering
	\vspace{-5pt}
	\includegraphics[width=3.4in]{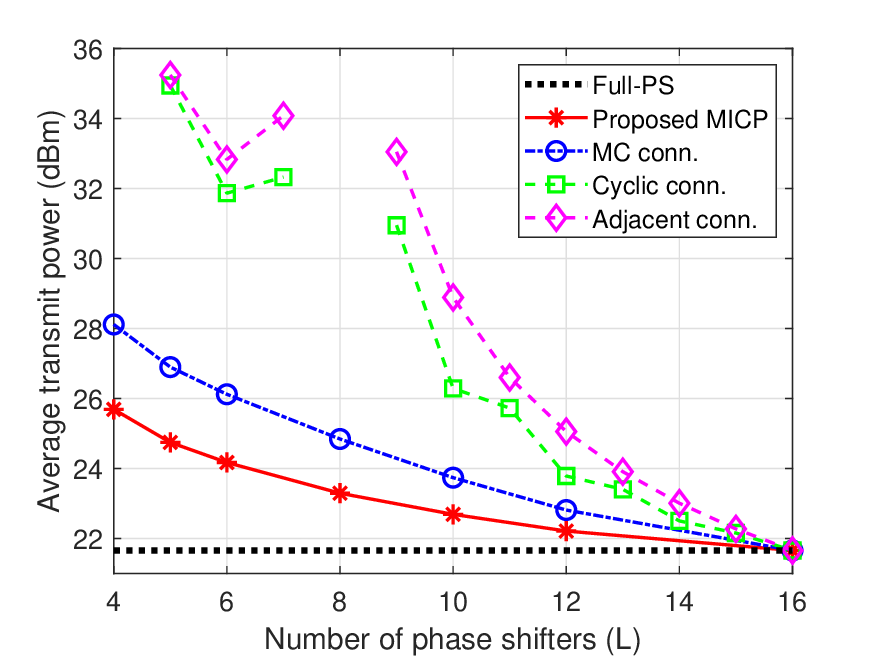}
	\vspace{-5pt}
	\caption{Average transmit power versus the number of phase shifters in the single-RF-chain system.}
	\label{Fig3}
	\includegraphics[width=3.4in]{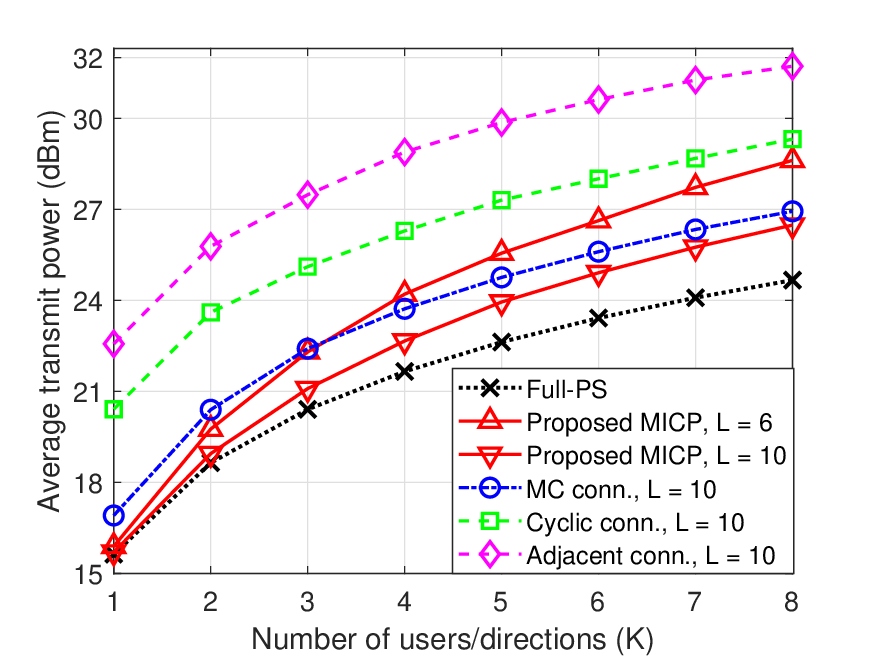}
	\vspace{-5pt}
	\caption{Average transmit power versus the number of users/directions in the single-RF-chain system.}
	\label{Fig4}
\end{figure}

Fig.~\ref{Fig3} illustrates the average transmit power versus the number of PSs $L$ for $K=4$. As observed, \textit{Proposed MICP} significantly outperforms all sparse-PS benchmarks and approaches the \textit{Full-PS} reference benchmark as $L$ increases. For example, at $L=10$, it reduces the PS count by $37.5\%$ compared with the full-PS architecture, while only increasing transmit power from about $21.7$~dBm to $22.7$~dBm. In contrast, \textit{Cyclic/Adjacent conn.} schemes exhibit evident performance fluctuations, and their results at $L\in\{4,8\}$ are not plotted because the required powers exceed $40$~dBm. This degradation stems from their deterministic PS connections: antennas sharing one PS may have steering phases that nearly cancel for some directions, producing deep nulls in the beampattern. Once a candidate AoD falls close to such a null, the array gain collapses, and a large transmit power is required to satisfy the QoS constraint, the underlying mechanism of which is further discussed in the analysis of Fig~\ref{Fig6}. The \textit{Proposed MICP} avoids this problem by optimizing PS connections according to the candidate users/directions, thereby preventing unfavorable antenna groupings.\looseness=-1

Fig.~\ref{Fig4} depicts the average transmit power versus the number of candidate users/directions $K$. Here, \textit{Proposed MICP} is evaluated with $L\in\{6,10\}$, while the other sparse-PS benchmarks adopt $L=10$. As expected, the transmit power consistently increases with $K$ for all schemes, since supporting more candidate users under a shared static connection matrix inevitably compromises the array gain for individual users, rendering the required power increasingly dominated by the worst-case direction. Among the schemes with $L=10$, \textit{Proposed MICP} achieves the lowest transmit power, closely approaching the \textit{Full-PS} reference benchmark. Remarkably, even with only $L=6$, the proposed design still outperforms the deterministic \textit{Cyclic/Adjacent conn.} schemes employing $L=10$ PSs, and surpasses $\textit{MC conn.}$ with $L=10$ when $K\le3$. This indicates that an optimized sparse-PS connection yields superior beamforming performance than simply deploying more PSs in a deterministic or random manner, highlighting the importance of connection optimization in sparse-PS architectures.\looseness=-1

\vspace{-3pt}
\subsection{Multi-RF-Chain System} Next, we study the general multi-RF-chain multiuser scenario, where the BS is equipped with $R=8$ RF-chains and serves $K$ candidate users/directions simultaneously. The ULA comprises $N=64$ antennas and is partitioned into $R$ sub-arrays, each consisting of $M=N/R=8$ antennas. The results obtained by \textit{Algorithm \ref{alg_profiledS_qosmm_new}} proposed in Section \ref{sec4} are designated as \textit{Proposed QoS-MM}, while the number of random connection matrices for \textit{MC conn.} and the penalty factor for problem (P3.3) are set to $r_\mathrm{a}=100$ and $\rho=10^4$, respectively.

\begin{figure}[t]
	\centering
	\vspace{-5pt}
	\includegraphics[width=3.4in]{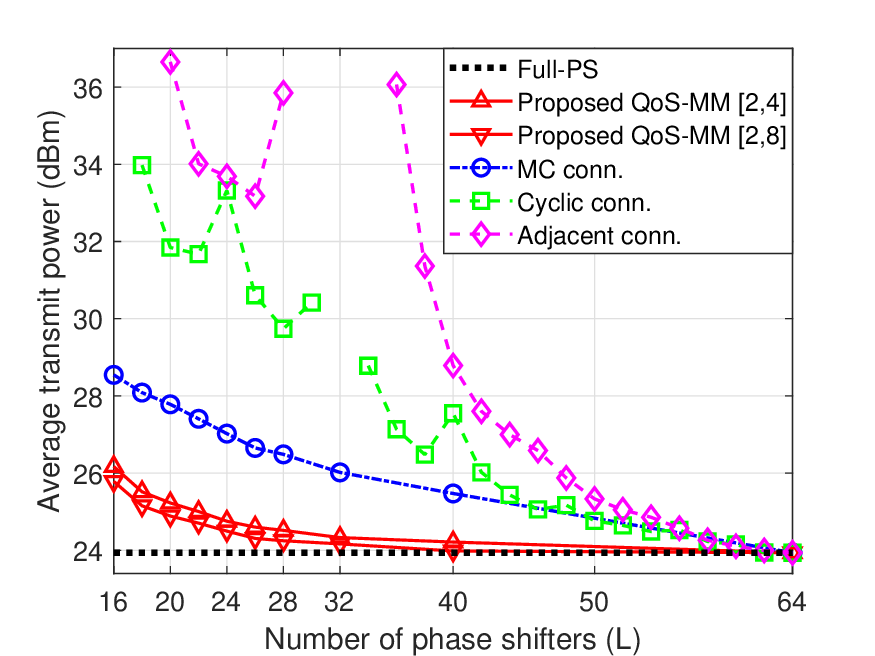}
	\vspace{-5pt}
	\caption{Average transmit power versus the number of phase shifters in the multi-RF-chain multiuser scenario.}
	\label{Fig6}
	\includegraphics[width=3.4in]{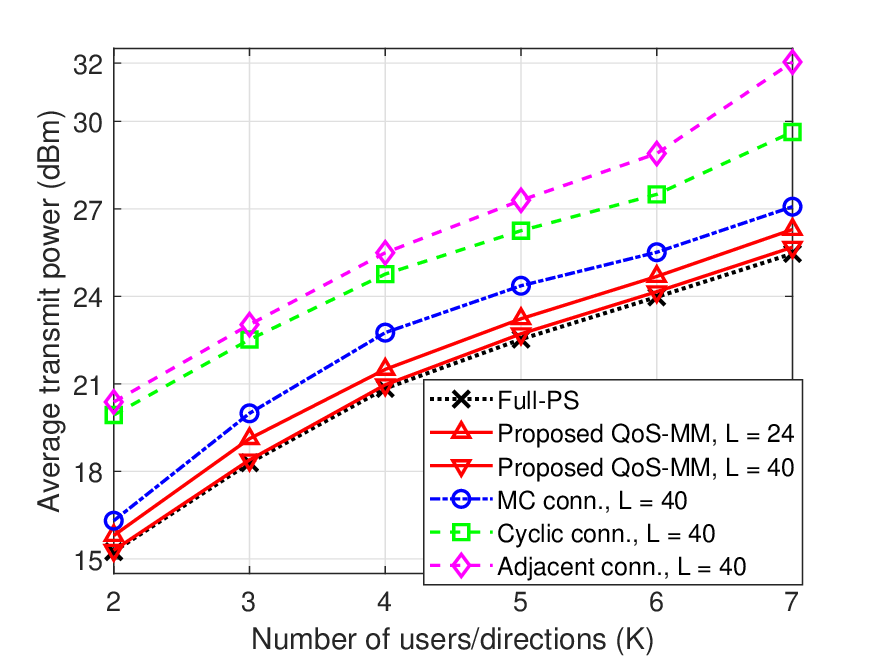}
	\vspace{-5pt}
	\caption{Average transmit power versus the number of users/directions in the multi-RF-chain system.}
	\label{Fig7}
\end{figure}

Fig.~\ref{Fig6} illustrates the average transmit power versus the total number of PSs $L$ for $K=6$, where $[2,4]$ and $[2,8]$ in \textit{Proposed QoS-MM} denote the Hamming-distance range $[\Delta_{\min},\Delta_{\max}]$ used in problem (P3.3). The proposed designs achieve significant power savings over all sparse-PS benchmarks and approach the \textit{Full-PS} reference benchmark as $L$ grows. Specifically, \textit{Proposed QoS-MM} with range $[2,8]$ requires $24.5$~dBm transmit power at $L=24$, which is about $2.5$~dB lower than that of \textit{MC conn.} and only around $0.5$~dB higher than the \textit{Full-PS} reference benchmark, while dramatically reducing the PS count by $62.5\%$. Moreover, the virtually identical curves for the $[2,4]$ and $[2,8]$ configurations demonstrate that exploring a tight Hamming-distance neighborhood is sufficient to identify near-optimal PS connections. This indicates that, even in complex scenarios with multi-stream transmissions and multiuser interference, a meticulously optimized sparse-PS connection can preserve most of the beamforming and interference-management capabilities of the full-PS architecture with substantially fewer PSs. Similar to Fig.~\ref{Fig3}, the deterministic \textit{Cyclic/Adjacent conn.} schemes exhibit severe performance fluctuations, and several points in the vicinity of $L$ for which $N/L$ is an integer are omitted because the required powers drastically exceed $40$~dBm. For \textit{Adjacent conn.}, each PS controls a contiguous group of $N/L$ antennas, whose beampattern contains deep direction-dependent nulls; for \textit{Cyclic conn.}, $N/L$ periodically spaced antennas connected to the same PS form sparse sub-arrays with pronounced grating nulls. Indeed, such beampattern defects may either align nulls with desired signal directions or create unsuppressed grating lobes toward interfering users, forcing the system to compensate for the resulting loss via elevated transmit power. In contrast, \textit{Proposed QoS-MM} mitigates this issue by jointly optimizing the PS connection and digital/analog beamformers, thereby avoiding rigid connection patterns that severely restrict the spatial DoFs.\looseness=-1


In Fig.~\ref{Fig7}, we further investigate the relationship between the average transmit power and the number of candidate users/directions $K$, with the Hamming-distance range of \textit{Proposed QoS-MM} set to $[2,6]$. As observed, while the transmit power naturally increases with $K$ across all schemes, the curve of \textit{Proposed QoS-MM} with $L=40$ almost overlaps with the \textit{Full-PS} reference benchmark for all tested values of $K$. This implies that, under the proposed design, employing more PSs beyond $L=40$ provides only marginal performance gains, regardless of the user scale. Furthermore, even with a reduced total PS budget of $L=24$, \textit{Proposed QoS-MM} consistently outperforms \textit{MC conn.} with $L=40$ over the entire range of $K$. Such a comparison highlights that an optimized connection topology using fewer PSs can be more effective than random connections with more PSs. These findings collectively confirm that the proposed sparse-PS connection design maintains excellent beamforming performance as the user scale expands, while achieving substantial PS hardware cost reductions.\looseness=-1

\section{Conclusion}\label{sec6}
This paper proposed a novel static-connection hybrid beamforming architecture with a reduced number of PSs to satisfy the stringent practical demand for ultra-low hardware cost and implementation complexity. By progressing from a single-RF-chain setup to a general multi-RF-chain system, this study provided novel analytical insights and pragmatic optimization algorithms--based on mixed-integer convex program and QoS-majorization-minimization, respectively--for the proposed sparse-PS static-connection architecture. In the single-RF-chain scenario, we revealed that preserving array gain depends not merely on the PS count, but also critically on grouping antennas with compatible steering phases. In the multi-RF-chain scenario, the proposed QoS-MM framework showed that jointly optimizing the discrete PS connections and continuous beamformers effectively recovers spatial DoFs despite a limited hardware budget. Numerical results affirmed that the proposed designs substantially reduce the PS count while maintaining near-full-PS performance and consistently outperforming deterministic or random connection schemes. Crucially, this work establishes sparse-PS hybrid beamforming fundamentally as a topology optimization problem rather than a simple hardware reduction attempt: with carefully optimized static connections, fewer PSs can deliver strong beamforming and interference-management capabilities, offering an attractive hardware-performance trade-off for cost-sensitive large-scale antenna systems.\looseness=-1

\appendices
\section{Proof of Lemma \ref{Lemma1}}\label{appendix_A}
We prove \textit{Lemma \ref{Lemma1}} by verifying $\sum\nolimits_{\bar{l}\in\bar{\mathcal{L}}}\left(\left\lfloor\frac{N-\bar{l}}{\bar{L}-1}\right\rfloor+1\right)=N$. Specifically, the above equation is equivalent to proving
\vspace{-0pt}
\begin{equation}\label{eq_lemma1_proof1}
	\sum\nolimits_{\bar{l}=1}^{\bar{L}-1}\left\lfloor\frac{N-\bar{l}}{\bar{L}-1}\right\rfloor=N-\bar{L}+1.
\end{equation}
To prove \eqref{eq_lemma1_proof1}, let $\bar{\kappa},\bar{r}\in\mathbb{N}$ satisfy
\begin{equation}\label{eq_lemma1_proof2}
	N=\big(\bar{L}-1\big)\bar{\kappa}+\bar{r},~0\le{\bar{r}}<\bar{L}-1,
\end{equation}
i.e., $\bar{\kappa}=\big\lfloor N/(\bar{L}-1)\big\rfloor$ and $\bar{r}$ is the remainder. If $1\le\bar{l}\le{\bar{r}}$, then\looseness=-1
\begin{align}
	&N-\bar{l}=\big(\bar{L}-1\big)\bar{\kappa}+\bar{r}-\bar{l},~0\le{\bar{r}}-\bar{l}<\bar{L}-1,\nonumber\\
	\Rightarrow~&\left\lfloor\frac{N-\bar{l}}{\bar{L}-1}\right\rfloor=\left\lfloor\bar{\kappa}+\frac{\bar{r}-\bar{l}}{\bar{L}-1}\right\rfloor=\bar{\kappa},
\end{align}
since $\big(\bar{r}-\bar{l}\big)/\big(\bar{L}-1\big)\in[0,1)$. If $\bar{r}<\bar{l}\le\bar{L}-1$, then
\begin{align}
	&N-\bar{l}=\big(\bar{L}-1\big)\bar{\kappa}+\bar{r}-\bar{l},~1-\bar{L}\le{\bar{r}}-\bar{l}<0,\nonumber\\
	\Rightarrow~&\left\lfloor\frac{N-\bar{l}}{\bar{L}-1}\right\rfloor=\left\lfloor\bar{\kappa}+\frac{\bar{r}-\bar{l}}{\bar{L}-1}\right\rfloor=\bar{\kappa}-1,
\end{align}
since $\big(\bar{r}-\bar{l}\big)/\big(\bar{L}-1\big)\in[-1,0)$. Therefore, the left-hand side of \eqref{eq_lemma1_proof1} can be rewritten as
\begin{align}
	&\sum\nolimits_{\bar{l}=1}^{\bar{L}-1}\left\lfloor\frac{N-\bar{l}}{\bar{L}-1}\right\rfloor=\sum\nolimits_{\bar{l}=1}^{\bar{r}}\left\lfloor\frac{N-\bar{l}}{\bar{L}-1}\right\rfloor+\sum\nolimits_{\bar{l}=\bar{r}+1}^{\bar{L}-1}\left\lfloor\frac{N-\bar{l}}{\bar{L}-1}\right\rfloor\nonumber\\
	&=\bar{r}\bar{\kappa}+\big(\bar{L}-1-\bar{r}\big)\big(\bar{\kappa}-1\big)\overset{\eqref{eq_lemma1_proof2}}{=}N-\bar{L}+1,
\end{align}
which completes the proof.

\section{Proof of Proposition \ref{Proposition1}}\label{appendix_B}
We define the maximum achievable array gain as
\begin{equation}
	G^\star(L)\triangleq\max_{\mathbf{\bar{C}},\,\bm{\bar{\phi}}}|\mathbf{a}(\theta)^H\mathbf{\bar{C}}\bm{\bar{\phi}}|.
\end{equation}
To prove the proposition, it suffices to demonstrate that $G^\star(L+1)\ge G^\star(L)$, for any $L<N$.
Let $\big\{\mathbf{\bar{C}}_\mathrm{opt},\bm{\bar{\phi}}_\mathrm{opt}\big\}$ be the optimal solution for the $L$-PS case, and let $\{\mathcal{N}_l\}_{l\in\mathcal{L}}$ denote the corresponding antenna groups. Since $L<N$ and each PS must be connected to at least one antenna, there exists at least one group containing more than one antenna. Select such a group and denote its index by $l_\mathrm{sele}$, i.e., $|\mathcal{N}_{l_\mathrm{sele}}|\ge 2$.	Choose one antenna $n_\mathrm{sele}\in\mathcal{N}_{l_\mathrm{sele}}$. Now construct a feasible solution $\big\{\mathbf{\bar{C}}',\bm{\bar{\phi}}'\big\}$ and the corresponding antenna groups $\{\mathcal{N}'_l\}_{l\in\mathcal{L}'}$ for the $(L+1)$-PS case as follows:
\begin{enumerate}
	\item Retain the first $L$ phase shifts the same as those in the $L$-PS optimum, i.e., $\bm{\bar{\phi}}'[l]=\bm{\bar{\phi}}_\mathrm{opt}[l]$, $l=1,\ldots,L$;
	\item Assign the newly added $(L+1)$-th PS the same phase shift as the $l_\mathrm{sele}$-th PS, i.e., $\bm{\bar{\phi}}'[L+1]=\bm{\bar{\phi}}_\mathrm{opt}[l_\mathrm{sele}]$;
	\item Keep all antenna groups unchanged except that the $n_\mathrm{sele}$-th antenna is removed from $\mathcal{N}_{l_\mathrm{sele}}$ and assigned to the $(L+1)$-th PS, namely, $\mathcal{N}'_{L+1}=\{n_\mathrm{sele}\}$, $\mathcal{N}'_{l_\mathrm{sele}}=\mathcal{N}_{l_\mathrm{sele}}\!\!\setminus\!\!\{n_\mathrm{sele}\}$, and $\mathcal{N}'_l=\mathcal{N}_l$, $\forall\,l\ne l_\mathrm{sele}$.
\end{enumerate}
Under the above construction, the resulting array gain is
\begin{align}
	&\hspace{3pt}|\mathbf{a}(\theta)^H\mathbf{\bar{C}}'\bm{\bar{\phi}}'|=\bigg|\sum\nolimits_{l=1}^{L+1}\bm{\bar{\phi}}'[l]\sum\nolimits_{n\in\mathcal{N}'_l}e^{j(n-1)\pi\sin(\theta)}\bigg|\nonumber\\
	&\hspace{-10pt}=\bigg|\sum\nolimits_{l=1,l\neq l_\mathrm{sele}}^{L}\bm{\bar{\phi}}_\mathrm{opt}[l]\sum\nolimits_{n\in\mathcal{N}_l}e^{j(n-1)\pi\sin(\theta)}+\bigg.\nonumber\\
	&\hspace{2pt}\Bigg.\bm{\bar{\phi}}_\mathrm{opt}[l_\mathrm{sele}]\hspace{-18pt}\sum_{n\in\mathcal{N}_{l_\mathrm{sele}}\hspace{-2pt}\setminus\{n_\mathrm{sele}\}}\hspace{-20pt}e^{j(n-1)\pi\sin(\theta)}\!+\!\bm{\bar{\phi}}_\mathrm{opt}[l_\mathrm{sele}]e^{j(n_\mathrm{sele}-1)\pi\sin(\theta)}\Bigg|\nonumber\\
	&\hspace{-10pt}=\left|\sum_{l=1}^{L}\bm{\bar{\phi}}_\mathrm{opt}[l]\!\sum_{n\in\mathcal{N}_l}\!e^{j(n-1)\pi\sin(\theta)}\right|=\big|\mathbf{a}(\theta)^H\mathbf{\bar{C}}_\mathrm{opt}\bm{\bar{\phi}}_\mathrm{opt}\big|.
\end{align}
Since $\big\{\mathbf{\bar{C}}',\bm{\bar{\phi}}'\big\}$ is a feasible solution for the $(L+1)$-PS case, while $G^\star(L+1)$ is its optimal value, we have
\begin{equation}
	G^\star(L+1)\ge|\mathbf{a}(\theta)^H\mathbf{\bar{C}}'\bm{\bar{\phi}}'|=G^\star(L),
\end{equation}
which completes the proof.

\section{Derivation of $\nabla_{\mathbf{\bar{C}}}f_k(\mathbf{\bar{C}})$}\label{appendix_C}
First, we express the steering vector in terms of $\mathbf{a}(\theta_k)=\mathbf{a}_k^\mathrm{R}+j\mathbf{a}_k^\mathrm{I}$, $k\in\mathcal{K}$, where $\mathbf{a}_k^\mathrm{R}\triangleq\Re\{\mathbf{a}(\theta_k)\}$ and $\mathbf{a}_k^\mathrm{I}\triangleq\Im\{\mathbf{a}(\theta_k)\}$. Then, $f_k\big(\mathbf{\bar{C}}\big)$ can be rewritten as
\begin{equation}
	f_k\big(\mathbf{\bar{C}}\big)=\sum\nolimits_{l\in\mathcal{L}}\!\sqrt{\big(\big(\mathbf{a}_k^\mathrm{R}\big)^T\mathbf{\bar{C}}[:,l]\big)^2\!\!+\!\big(\big(\mathbf{a}_k^\mathrm{I}\big)^T\mathbf{\bar{C}}[:,l]\big)^2}.
\end{equation}
For each entry $\mathbf{\bar{C}}[n,l]$ viewed as a continuous variable, with $n\in\mathcal{N}$ and $l\in\mathcal{L}$, the $(n,l)$-th entry of gradient matrix $\nabla_{\mathbf{\bar{C}}}f_k\big(\mathbf{\bar{C}}\big)$ can be obtained by calculating
\vspace{0pt}
\begin{align}
	\nabla_{\mathbf{\bar{C}}}f_k\big(\mathbf{\bar{C}}\big)[n,l]&=\frac{\partial f_k\big(\mathbf{\bar{C}}\big)}{\partial\mathbf{\bar{C}}[n,l]}\nonumber\\
	&\hspace{-20pt}=\frac{\partial}{\partial\mathbf{\bar{C}}[n,l]}\sqrt{\big(\big(\mathbf{a}_k^\mathrm{R}\big)^T\mathbf{\bar{C}}[:,l]\big)^2+\big(\big(\mathbf{a}_k^\mathrm{I}\big)^T\mathbf{\bar{C}}[:,l]\big)^2}\nonumber\\
	&\hspace{-20pt}=\frac{\big(\mathbf{a}_k^\mathrm{R}\big)^T\mathbf{\bar{C}}[:,l]\mathbf{a}_k^\mathrm{R}[n]+\big(\mathbf{a}_k^\mathrm{I}\big)^T\mathbf{\bar{C}}[:,l]\mathbf{a}_k^\mathrm{I}[n]}{\sqrt{\big(\big(\mathbf{a}_k^\mathrm{R}\big)^T\mathbf{\bar{C}}[:,l]\big)^2+\big(\big(\mathbf{a}_k^\mathrm{I}\big)^T\mathbf{\bar{C}}[:,l]\big)^2}}.
\end{align}

\section{Derivation of Closed-Form Expression of $\beta_k$}\label{appendix_D}
We first establish a tractable upper bound on $\lambda_{\max}(\mathbf{Q}_k)$. For any $\mathbf{b}_{k,i}\in\mathbb C^{LN\times1}$ and $\mathbf x\in\mathbb R^{LN\times1}$, since the matrix $\Re\big\{\mathbf b_{k,i}^{*}\mathbf b_{k,i}^{T}\big\}$ is real symmetric positive semidefinite, applying the Cauchy-Schwarz inequality yields
\begin{equation}
	\mathbf x^T\Re\big\{\mathbf{b}_{k,i}^*\mathbf{b}_{k,i}^T\big\}\mathbf x=\big|\mathbf b_{k,i}^{T}\mathbf x\big|^2\le\|\mathbf b_{k,i}\|_2^2\|\mathbf x\|_2^2.
\end{equation}
By the Rayleigh quotient, this inequality directly implies that
\begin{equation}
	\lambda_{\max}\big(\Re\big\{\mathbf{b}_{k,i}^*\mathbf{b}_{k,i}^T\big\}\big)\le\|\mathbf b_{k,i}\|_2^2.
\end{equation}
Since the desired-signal term $-\Re\big\{\mathbf{b}_{k,k}^*\mathbf{b}_{k,k}^T\big\}$ is real symmetric negative semidefinite, its inclusion does not increase the maximum eigenvalue of the sum. Consequently, from the definition of $\mathbf{Q}_k$, we obtain
\begin{align}
	\lambda_{\max}(\mathbf Q_k)
	&=\lambda_{\max}\!\left(
	\gamma_k\!\!\sum\nolimits_{i\in\mathcal{K},i\ne k}\!\!\!\Re\big\{\mathbf{b}_{k,i}^*\mathbf{b}_{k,i}^T\big\}
	\!-\!\Re\big\{\mathbf{b}_{k,k}^*\mathbf{b}_{k,k}^T\big\}
	\!\right) \nonumber\\
	&\le\gamma_k
	\sum\nolimits_{i\in\mathcal{K},i\ne k}
	\lambda_{\max}\big(\Re\big\{\mathbf{b}_{k,i}^*\mathbf{b}_{k,i}^T\big\}\big) \nonumber\\
	&\le\gamma_k
	\sum\nolimits_{i\in\mathcal{K},i\ne k}
	\|\mathbf b_{k,i}\|_2^2,~k\in\mathcal{K}.\label{eq_beta_k}
\end{align}

Next, by recalling the definition $\mathbf{b}_{k,i}=\operatorname{vec}\Big(\bm\phi\big(\mathbf{\hat{H}}_k\mathbf{d}_i\big)^T\Big)$, $\|\mathbf{b}_{k,i}\|_2^2$ can be computed as
\begin{equation}\label{eq_beta_k_new}
	\|\mathbf b_{k,i}\|_2^2=\big\|\bm\phi\big(\mathbf{\hat{H}}_k\mathbf{d}_i\big)^T\big\|_\mathrm{F}^2=\|\bm\phi\|_2^2\big\|\mathbf{\hat{H}}_k\mathbf d_i\big\|_2^2\overset{\mathrm{(f)}}{=}L\big\|\mathbf{\hat{H}}_k\mathbf d_i\big\|_2^2,
\end{equation}
where (f) follows from $|\phi_l|=1,~l\in\mathcal L$. Substituting \eqref{eq_beta_k_new} into
\eqref{eq_beta_k} gives $\beta_k=\gamma_kL\sum_{i\in\mathcal{K},i\ne k}\big\|\mathbf{\hat{H}}_k\mathbf{d}_i\big\|_2^2,~k\in\mathcal{K}$. Therefore, we have
\begin{equation}
	\beta_k\mathbf{I}^{LN}-\mathbf{Q}_k\succeq\mathbf{0},~k\in\mathcal{K},
\end{equation}
which proves the validity of \eqref{eq_beta_k_closed_form}.

\section{Derivation of Linear MM Surrogates}\label{appendix_E}
We first derive the individual linear MM surrogate in
\eqref{eq_residual_surrogate_new}. Since \textbf{$\beta_k\mathbf I^{LN}-\mathbf Q_k\succeq\mathbf 0$}, for any feasible binary vector $\mathbf{\hat{c}}$ and current point $\mathbf{\hat{c}}^{[t]}$, it holds that
\vspace{0pt}
\begin{equation}\label{eq_MM_psd_start_proof}
	\big(\mathbf{\hat{c}}-\mathbf{\hat{c}}^{[t]}\big)^T\big(\beta_k\mathbf I^{LN}-\mathbf Q_k\big)\big(\mathbf{\hat{c}}-\mathbf{\hat{c}}^{[t]}\big)\ge0.
\end{equation}
Expanding \eqref{eq_MM_psd_start_proof} yields
\begin{align}\label{eq_MM_quadratic_bound_proof}
	\mathbf{\hat{c}}^T\mathbf Q_k\mathbf{\hat{c}}\le&\;\beta_k\mathbf{\hat{c}}^T\mathbf{\hat{c}}-2\beta_k\mathbf{\hat{c}}^T\mathbf{\hat{c}}^{[t]}+2\mathbf{\hat{c}}^T\mathbf Q_k\mathbf{\hat{c}}^{[t]}\nonumber\\
	&+\big(\mathbf{\hat{c}}^{[t]}\big)^T\big(\beta_k\mathbf I^{LN}-\mathbf Q_k\big)\mathbf{\hat{c}}^{[t]}.
\end{align}
Since $\mathbf{\hat{c}}$ is binary-valued, we have	$\mathbf{\hat{c}}^T\mathbf{\hat{c}}=(\mathbf 1^{LN})^T\mathbf{\hat{c}}$. Substituting this identity into
\eqref{eq_MM_quadratic_bound_proof} and introducing the constant term $\gamma_kM\sigma_k^2$ from \eqref{eq_Qk_single_def_new} gives
\begin{align}\label{eq_MM_residual_bound_proof}
	r_k(\mathbf{\hat{c}})\le&\,\big[\,2\mathbf Q_k\mathbf{\hat{c}}^{[t]}+\beta_k\big(\mathbf 1^{LN}-2\mathbf{\hat{c}}^{[t]}\big)\big]^T\mathbf{\hat{c}}\nonumber\\
	&+\big(\mathbf{\hat{c}}^{[t]}\big)^T\big(\beta_k\mathbf I^{LN}-\mathbf Q_k\big)\mathbf{\hat{c}}^{[t]}+\gamma_kM\sigma_k^2\nonumber\\
	=&\;\bm\xi_k^T\mathbf{\hat{c}}+\kappa_k=
	\hat r_k\big(\mathbf{\hat{c}}\mid\mathbf{\hat{c}}^{[t]}\big).
\end{align}
When $\mathbf{\hat{c}}=\mathbf{\hat{c}}^{[t]}$, the left-hand side of \eqref{eq_MM_psd_start_proof} is zero. Hence,
\begin{equation}
	\hat r_k\big(\mathbf{\hat{c}}^{[t]}\!\mid\!\mathbf{\hat{c}}^{[t]}\big)=r_k\big(\mathbf{\hat{c}}^{[t]}\big).
\end{equation}

Next, we derive the dual-weighted aggregate linear MM surrogate. Given that $\eta_k\ge0$, $\forall\,k\in\mathcal{K}$, summing \eqref{eq_MM_residual_bound_proof} over $k\in\mathcal{K}$ with weights $\{\eta_k\}_{k\in\mathcal{K}}$ yields
\begin{align}
	J_C(\mathbf{\hat{c}})&=\sum_{k\in\mathcal K}\eta_kr_k(\mathbf{\hat{c}})\nonumber\\
	&\le\sum_{k\in\mathcal K}\!\eta_k\hat r_k\big(\mathbf{\hat{c}}\!\mid\!\mathbf{\hat{c}}^{[t]}\big)\!=\bm\xi_C^T\mathbf{\hat{c}}+\kappa_C=\hat J_C\big(\mathbf{\hat{c}}\!\mid\!\mathbf{\hat{c}}^{[t]}\big).
\end{align}
Furthermore, with the definitions of $\mathbf Q_C$ and $\beta_C$, we obtain
\begin{equation}
	\beta_C\mathbf I^{LN}-\mathbf Q_C=\sum\nolimits_{k\in\mathcal K}\eta_k\big(\beta_k\mathbf I^{LN}-\mathbf Q_k\big)\succeq\mathbf 0.
\end{equation}
Since each individual surrogate is tight at $\mathbf{\hat{c}}^{[t]}$, their nonnegative weighted sum is also tight:
\begin{equation}
	\hat{J}_C\big(\mathbf{\hat{c}}^{[t]}\!\mid\!\mathbf{\hat{c}}^{[t]}\big)=J_C\big(\mathbf{\hat{c}}^{[t]}\big).
\end{equation}
This completes the derivation of the individual and aggregate linear MM surrogates.

\bibliographystyle{ieeetr}
\bibliography{bibfile}

\begin{thebibliography}{10}

\bibitem{Elbir_BF_suv}
A.~M. Elbir {\em et~al.}, ``Twenty-five years of advances in beamforming: From
  convex and nonconvex optimization to learning techniques,'' {\em IEEE Signal
  Process. Mag.}, vol.~40, pp.~118--131, Jun. 2023.

\bibitem{WanSJ_HB_GNN}
S.~Wan, Z.~Wang, and Y.~Zhou, ``Scalable hybrid beamforming for multi-user
  {MISO} systems: A graph neural network approach,'' {\em IEEE Trans. Wireless
  Commun.}, vol.~23, pp.~13694--13706, Oct. 2024.

\bibitem{WangHH_low_altitude_Mag}
H.~Wang {\em et~al.}, ``Reconfigurable airspace: Synergizing movable antenna
  and intelligent surface for low-altitude {ISAC} networks,'' {\em {\rm 2025},
  arXiv: 2511.10310}.

\bibitem{Molisch_HB_suv}
A.~F. Molisch, V.~V. Ratnam, S.~Han, Z.~Li, S.~L.~H. Nguyen, L.~Li, and
  K.~Haneda, ``Hybrid beamforming for massive {MIMO}: A survey,'' {\em IEEE
  Commun. Mag.}, vol.~55, pp.~134--141, Sep. 2017.

\bibitem{Albreem_BF_suv}
M.~A. Albreem {\em et~al.}, ``Overview of precoding techniques for massive
  {MIMO},'' {\em IEEE Access}, vol.~9, pp.~60764--60801, 2021.

\bibitem{WangHH_delay_MA}
H.~Wang {\em et~al.}, ``Throughput maximization for movable antenna systems
  with movement delay consideration,'' {\em IEEE Trans. Wireless Commun.},
  vol.~25, pp.~883--899, 2026.

\bibitem{Sohrabi_HB_suv}
F.~Sohrabi and W.~Yu, ``Hybrid digital and analog beamforming design for
  large-scale antenna arrays,'' {\em IEEE J. Sel. Top. Signal Process.},
  vol.~10, pp.~501--513, Apr. 2016.

\bibitem{YuXH_HB_switch_1}
X.~Yu, J.-C. Shen, J.~Zhang, and K.~B. Letaief, ``Alternating minimization
  algorithms for hybrid precoding in millimeter wave {MIMO} systems,'' {\em
  IEEE J. Sel. Top. Signal Process.}, vol.~10, pp.~485--500, Apr. 2016.

\bibitem{GaoXY_HB_EE}
X.~Gao {\em et~al.}, ``Energy-efficient hybrid analog and digital precoding for
  mmwave {MIMO} systems with large antenna arrays,'' {\em IEEE J. Sel. Areas
  Commun.}, vol.~34, pp.~998--1009, Apr. 2016.

\bibitem{Ahmed_HB_suv}
I.~Ahmed {\em et~al.}, ``A survey on hybrid beamforming techniques in {5G}:
  Architecture and system model perspectives,'' {\em IEEE Commun. Surv. Tuts.},
  vol.~20, no.~4, pp.~3060--3097, 2018.

\bibitem{YuXH_HB_switch_2}
X.~Yu, J.~Zhang, and K.~B. Letaief, ``A hardware-efficient analog network
  structure for hybrid precoding in millimeter wave systems,'' {\em IEEE J.
  Sel. Top. Signal Process.}, vol.~12, pp.~282--297, May 2018.

\bibitem{ZhangZL_HB_sub}
Z.~Zhang, X.~Wu, and D.~Liu, ``Joint precoding and combining design for hybrid
  beamforming systems with subconnected structure,'' {\em IEEE Sys. J.},
  vol.~14, pp.~184--195, Mar. 2020.

\bibitem{Bogale_HB_RF_PS}
T.~E. Bogale {\em et~al.}, ``On the number of {RF} chains and phase shifters,
  and scheduling design with hybrid analog–digital beamforming,'' {\em IEEE
  Trans. Wireless Commun.}, vol.~15, pp.~3311--3326, May 2016.

\bibitem{Payami_HB_PS}
S.~Payami, M.~Ghoraishi, and M.~Dianati, ``Hybrid beamforming for large antenna
  arrays with phase shifter selection,'' {\em IEEE Trans. Wireless Commun.},
  vol.~15, pp.~7258--7271, Nov. 2016.

\bibitem{Payami_HB_PS_switch}
S.~Payami {\em et~al.}, ``Phase shifters versus switches: An energy efficiency
  perspective on hybrid beamforming,'' {\em IEEE Wireless Commun. Lett.},
  vol.~8, pp.~13--16, Feb. 2019.

\bibitem{ZhuXD_HB_adaptive}
X.~Zhu {\em et~al.}, ``Adaptive hybrid precoding for multiuser massive
  {MIMO},'' {\em IEEE Commun. Lett.}, vol.~20, pp.~776--779, Apr. 2016.

\bibitem{LiHY_HB_low_resolution_PS_dynamic_array}
H.~Li, M.~Li, and Q.~Liu, ``Hybrid beamforming with dynamic subarrays and
  low-resolution {PSs} for mmwave {MU-MISO} systems,'' {\em IEEE Trans.
  Commun.}, vol.~68, pp.~602--614, Jan. 2020.

\bibitem{GaoH_HB_low_reso_DAC}
H.~Gao, D.~Liu, and Z.~Zhang, ``Low-resolution {DACs} aided partially-connected
  hybrid precoding design for {MU-MISO} system,'' {\em IEEE Commun. Lett.},
  vol.~27, pp.~2767--2771, Oct. 2023.

\bibitem{Alkhateeb_HB_fixed_PS_switch}
A.~Alkhateeb {\em et~al.}, ``Massive {MIMO} combining with switches,'' {\em
  IEEE Wireless Commun. Lett.}, vol.~5, pp.~232--235, Jun. 2016.

\bibitem{Bogale_HB_fixed_PS_switch}
T.~E. Bogale, L.~B. Le, and X.~Wang, ``Hybrid analog-digital channel estimation
  and beamforming: Training-throughput tradeoff,'' {\em IEEE Trans. Commun.},
  vol.~63, pp.~5235--5249, Dec. 2015.

\bibitem{Ardah_HB_PS_switch}
K.~Ardah {\em et~al.}, ``A unifying design of hybrid beamforming architectures
  employing phase shifters or switches,'' {\em IEEE Trans. Veh. Technol.},
  vol.~67, pp.~11243--11247, Nov. 2018.

\bibitem{Méndez_HB_PS_switch}
R.~Méndez-Rial {\em et~al.}, ``Hybrid {MIMO} architectures for millimeter wave
  communications: Phase shifters or switches?,'' {\em IEEE Access}, vol.~4,
  pp.~247--267, 2016.

\bibitem{JinJN_HB_dynamic_array}
J.~Jin, C.~Xiao, W.~Chen, and Y.~Wu, ``Channel-statistics-based hybrid
  precoding for millimeter-wave {MIMO} systems with dynamic subarrays,'' {\em
  IEEE Trans. Commun.}, vol.~67, pp.~3991--4003, Jun. 2019.

\bibitem{Park_HB_dynamic_array}
S.~Park, A.~Alkhateeb, and R.~W. Heath, ``Dynamic subarrays for hybrid
  precoding in wideband mmwave {MIMO} systems,'' {\em IEEE Trans. Wireless
  Commun.}, vol.~16, pp.~2907--2920, May 2017.

\bibitem{Ratnam_HB_sele}
V.~V. Ratnam {\em et~al.}, ``Hybrid beamforming with selection for multiuser
  massive {MIMO} systems,'' {\em IEEE Trans. Signal Process.}, vol.~66,
  pp.~4105--4120, Aug. 2018.

\bibitem{HuYB_HB_sub}
Y.~Hu, H.~Qian, K.~Kang, X.~Luo, and H.~Zhu, ``Joint precoding design for
  sub-connected hybrid beamforming system,'' {\em IEEE Trans. Wireless
  Commun.}, vol.~23, pp.~1199--1212, Feb. 2024.

\bibitem{Rupakula_HB_PS_num}
B.~Rupakula {\em et~al.}, ``Limited scan-angle phased arrays using randomly
  grouped subarrays and reduced number of phase shifters,'' {\em IEEE Trans.
  Antennas and Propag.}, vol.~68, pp.~70--80, Jan. 2020.

\bibitem{Valle_HB_PS_num}
J.~L. Valle {\em et~al.}, ``Reduction of phase shifters in planar phased arrays
  using novel random subarray techniques,'' {\em Appl. Sci.}, vol.~14, Jul.
  2024.

\bibitem{Balanis_antenna}
C.~A. Balanis, {\em Antenna theory: analysis and design}.
\newblock John wiley \& sons, 2016.

\bibitem{Aldubaikhy_fixed_access}
K.~Aldubaikhy {\em et~al.}, ``{mmWave} {IEEE} 802.11ay for {5G} fixed wireless
  access,'' {\em IEEE Wireless Commun.}, vol.~27, pp.~88--95, Apr. 2020.

\bibitem{LiJ_train}
J.~Li {\em et~al.}, ``Mobility support for millimeter wave communications:
  Opportunities and challenges,'' {\em IEEE Commun. Surv. Tuts.}, vol.~24,
  no.~3, pp.~1816--1842, 2022.

\bibitem{Jabbar_industry}
A.~Jabbar {\em et~al.}, ``Millimeter-wave smart antenna solutions for {URLLC}
  in industry 4.0 and beyond,'' {\em Sensors}, vol.~22, Mar. 2022.

\bibitem{TanJR_V2X}
J.~Tan {\em et~al.}, ``Beam alignment in {mmWave} {V2X} communications: A
  survey,'' {\em IEEE Commun. Surv. Tuts.}, vol.~26, no.~3, pp.~1676--1709,
  2024.

\bibitem{Giordani_beam_3GPP}
M.~Giordani, M.~Polese, A.~Roy, D.~Castor, and M.~Zorzi, ``A tutorial on beam
  management for {3GPP} {NR} at mmwave frequencies,'' {\em IEEE Commun. Surv.
  Tuts.}, vol.~21, no.~1, pp.~173--196, 2019.

\bibitem{Alkhateeb_codebook}
A.~Alkhateeb, G.~Leus, and R.~W. Heath, ``Limited feedback hybrid precoding for
  multi-user millimeter wave systems,'' {\em IEEE Trans. Wireless Commun.},
  vol.~14, pp.~6481--6494, Nov. 2015.

\bibitem{Qurratulain_Beam}
M.~Qurratulain~Khan, A.~Gaber, P.~Schulz, and G.~Fettweis, ``Machine learning
  for millimeter wave and terahertz beam management: A survey and open
  challenges,'' {\em IEEE Access}, vol.~11, pp.~11880--11902, 2023.

\bibitem{XueQ_Beam}
Q.~Xue, C.~Ji, S.~Ma, J.~Guo, Y.~Xu, Q.~Chen, and W.~Zhang, ``A survey of beam
  management for {mmWave} and {THz} communications towards {6G},'' {\em IEEE
  Commun. Surv. Tuts.}, vol.~26, no.~3, pp.~1520--1559, 2024.

\bibitem{Alrabeiah_Beam}
M.~Alrabeiah and A.~Alkhateeb, ``Deep learning for mmwave beam and blockage
  prediction using sub-6 {GHz} channels,'' {\em IEEE Trans. Commun.}, vol.~68,
  pp.~5504--5518, Sep. 2020.

\bibitem{Grant_CVX_software}
M.~Grant and S.~Boyd, ``{CVX}: {MATLAB} software for disciplined convex
  programming, version 2.2,'' Jan. 2020.

\bibitem{WangHH_interference_MA}
H.~Wang, Q.~Wu, and W.~Chen, ``Movable antenna enabled interference network:
  Joint antenna position and beamforming design,'' {\em IEEE Wireless Commun.
  Lett.}, vol.~13, pp.~2517--2521, Sep. 2024.

\bibitem{Schubert_up_downlink_duality}
M.~Schubert and H.~Boche, ``Solution of the multiuser downlink beamforming
  problem with individual {SINR} constraints,'' {\em IEEE Trans. Veh.
  Technol.}, vol.~53, pp.~18--28, Jan. 2004.

\bibitem{SunY_surrogate_function}
Y.~Sun, P.~Babu, and D.~P. Palomar, ``Majorization-minimization algorithms in
  signal processing, communications, and machine learning,'' {\em IEEE Trans.
  Signal Process.}, vol.~65, pp.~794--816, Feb. 2017.

\bibitem{Horn_matrix_2012}
R.~A. Horn and C.~R. Johnson, {\em Matrix Analysis}.
\newblock Cambridge, U.K.: Cambridge University Press, 2nd~ed., 2012.

\bibitem{Polik_interior_point_book}
I.~P{\'o}lik and T.~Terlaky, {\em Interior Point Methods for Nonlinear
  Optimization}.
\newblock Berlin, Germany: Springer, 2010.

\bibitem{WuYF_RIS_GBD}
Y.~Wu {\em et~al.}, ``Globally optimal resource allocation design for discrete
  phase shift {IRS}-assisted multiuser networks with perfect and imperfect
  {CSI},'' {\em IEEE Trans. Wireless Commun.}, vol.~24, pp.~1306--1324, 2025.

\bibitem{ZhangYP_mmChannel}
Y.~Zhang and C.~You, ``{SWIPT} in mixed near- and far-field channels: Joint
  beam scheduling and power allocation,'' {\em IEEE J. Sel. Areas Commun.},
  vol.~42, pp.~1583--1597, Jun. 2024.

\bibitem{Rappaport_mmChannel}
T.~S. Rappaport {\em et~al.}, ``Overview of millimeter wave communications for
  fifth-generation ({5G}) wireless networks—with a focus on propagation
  models,'' {\em IEEE Trans. Antennas and Propag.}, vol.~65, pp.~6213--6230,
  Dec. 2017.

\end{thebibliography}
\end{document}